\documentclass[a4paper,twocolumn,11pt,accepted=2021-01-03]{quantumarticle}
\pdfoutput=1
\usepackage[utf8]{inputenc}
\usepackage[english]{babel}
\usepackage[T1]{fontenc}
\usepackage[numbers,sort&compress]{natbib}
\usepackage{amsmath}
\usepackage{amsfonts}
\usepackage{tikz}

\usepackage{graphicx}
\usepackage{bm}
\usepackage{hyperref}

\usepackage{csquotes} 
\usepackage{physics} 
\usepackage{dsfont} 
\usepackage{mathtools} 
\usepackage{amsthm} 
\usepackage{amsbsy} 

\newcommand{\Hilb}[1]{\mathcal{H}^{#1}}
\newcommand{\LinOp}[1]{\mathcal{L}\left( #1 \right)}

\newcommand{\MapX}[2]{ \frac{\mathds{1}^{#1}}{d_{#1}} \otimes \Tr_{#1}\left[ #2 \right] }
\newcommand{\TrX}[2]{\mathrm{Tr}_{#1}\left[ #2 \right]}
\newcommand{\InProd}[2]{\left\langle #1 \left. ,\, #2 \right. \right\rangle}

\newcommand{\GGB}[2]{\sigma^{#1}_{#2}}
\newcommand{\Proj}[2]{\mathcal{P}^{#1}_{#2}}
\newcommand{\CompProj}[2]{\prescript{Q}{}{\mathcal{P}}^{#1}_{#2}}
\newcommand{\ProjOn}[3]{\mathcal{P}^{#1}_{#2}\left[#3\right]}

\newcommand{\Party}[3]{#1^{(#2)}_{#3}}
\newcommand{\Dep}[2]{\prescript{}{#1}{#2}}
\newcommand{\DepPar}[2]{\prescript{}{#1}{\left(#2\right)}}

\newtheorem{theo}{Theorem}
\newtheorem{defi}{Definition}
\newtheorem{lemm}{Lemma}

\newtheorem{coro}{Corollary}

\begin{document}

\title{The Multi-round Process Matrix}

\author{Timoth\'{e}e Hoffreumon}%
 \email{hoffreumon.timothee@ulb.ac.be}
 \affiliation{%
 Centre for Quantum Information and Communication (QuIC), {\'E}cole polytechnique de Bruxelles, CP 165, Universit\'e libre de Bruxelles, 1050 Brussels, Belgium.
}%
 \orcid{0000-0001-7014-1958}
\author{Ognyan Oreshkov}%
 \email{oreshkov@ulb.ac.be}
 \affiliation{%
 Centre for Quantum Information and Communication (QuIC), {\'E}cole polytechnique de Bruxelles, CP 165, Universit\'e libre de Bruxelles, 1050 Brussels, Belgium.
}%
 \orcid{0000-0002-9390-1064}

\maketitle

\begin{abstract}
We develop an extension of the process matrix (PM) framework for correlations between quantum operations with no causal order that allows multiple rounds of information exchange for each party compatibly with the assumption of well-defined causal order of events locally. We characterise the higher-order process describing such correlations, which we name the multi-round process matrix (MPM), and formulate a notion of causal nonseparability for it that extends the one for standard PMs. We show that in the multi-round case there are novel manifestations of causal nonseparability that are not captured by a naive application of the standard PM formalism: we exhibit an instance of an operator that is both a valid PM and a valid MPM, but is causally separable in the first case and can violate causal inequalities in the second case due to the possibility of using a side channel.
\end{abstract}

\section{\label{sec:Intro}Introduction}
There has recently been significant interest in quantum processes in which the operations performed by separate parties exhibit `indefinite causal order' \cite{Hardy2007,Chiribella2013,OCB2012,Araujo2014,Brukner2014,Procopio2015,Witness,OG2016,Feix2015,Baumeler_2016,Portmann2015,Feix_2016,Guerin2016,Maclean2017,Expe2017,Allen2017,Kissinger2017,Silva_2017,Abbott2017genuinely,Zych2016,Rubino2017,Goswami2018,Perinotti2016,Bisio2018,Wechs2018,Wei_2019,Procopio,tobar2020reversible}. A formal definition of this feature, termed causal nonseparability \cite{OCB2012,Witness,OG2016,Wechs2018}, has been given in the process matrix (PM) framework \cite{OCB2012}, which describes correlations between elementary quantum experiments, each defined by a pair of input and an output system, also referred to as a \textit{quantum node} \cite{Allen2017,barrett2019quantum,barrett2020cyclic}, over which an agent could apply different operations, without presuming the existence of global causal order among the separate operations but only the validity of causal quantum theory \cite{Chiribella_2011} for their local description. Causally nonseparable processes have been shown to accomplish informational tasks that are not possible via processes in which the operations are used in a defined order \cite{Chiribella2012,Araujo2014,Feix2015,Guerin2016,Procopio}. They have been conjectured to be potentially relevant in quantum gravity scenarios \cite{Hardy2007,Chiribella2013,OCB2012,Zych2016} where the causal structure of spacetime may be subject to quantum indefiniteness, as well as in the presence of closed timelike curves \cite{Chiribella2013,OCB2012,Brukner2014,Araujo_2017,Baumeler_2019,tobar2020reversible}. But some of these processes, such as the quantum SWITCH \cite{Chiribella2013} for which most of the examples of advantage have been found, also admit realisations within standard quantum mechanics on time-delocalized systems \cite{Oreshkov2018} via coherent control of the order of operations. This has been demonstrated in several experimental setups \cite{Procopio2015,Rubino2017,Expe2017,Maclean2017,Goswami2018,Wei_2019}, offering blueprints for possible applications. 

In view of developing potential applications of indefinite causal order, the standard PM framework appears limited by the fact that each party is assumed restricted to a single round of information exchange, where from the local causal perspective of that party \cite{Guerin_2018}, information is received once via the input system and subsequently sent out once via the output system. Practical communication protocols and distributed computing algorithms generally involve multiple rounds of information exchange between separate parties that can use local memory and condition the operations applied at a given time on information obtained at other times. It is therefore natural to ask whether the PM framework can be extended to allow for multiple rounds of communication and whether such an extension would contain any new possibilities that are not captured by the standard PM framework.

In this paper, we formulate a multi-round extension of the PM framework in which each party can perform operations over an ordered sequence of quantum nodes, assuming standard causality locally but not globally. The formalism is analogous to the PM formalism, except that the local operations are now generalised to quantum networks \cite{Chiribella2009}---the most general causal operations that can be applied over a given sequence of input and output systems, and which are described by the theory of quantum combs \cite{Chiribella2009}---while the process matrix is replaced by an operator that we call the \textit{multi-round process matrix} (MPM). The MPM is a specific higher-order process in the hierarchy of higher order processes classified by Perinotti and Bisio \cite{Perinotti2016,Bisio2018}. We derive handy necessary and sufficient conditions for an operator to be a valid MPM, which are expressed via a generalisation of the projector techniques introduced for PMs in Ref. \cite{Witness}. Given a set of parties and an ordered set of nodes for each party, a valid MPM is in particular a valid PM on all nodes, if each node is regarded as belonging to a separate party. However, it respects additional constraints that ensure the possibility of using side channels for causal communication between the nodes within the laboratory of each actual party. As we show, these constraints amount to the condition that the operator of a valid MPM is an affine combination of deterministic quantum combs that are compatible with the local orders of the nodes assumed for each different party. 

We further formulate a notion of causal separability in the MPM framework, building upon the one defined for PMs \cite{OG2016,Wechs2018}. The key insight behind this definition is that since every quantum comb can be implemented as a sequence of independent operations connected via side channels, by viewing these side channels as `external' to the parties' operations, we can reduce the problem of defining causal separability to that for standard PMs. Remarkably, we show that the possibility of using side channels can make a radical difference on the causal (non)separability of the process: we describe an example of a bipartite MPM with two nodes for Alice and one node for Bob, which is such that, when viewed as a standard PM on three nodes it is causally separable (and admitting a simple physical realisation), but when the possibility of using a side channel in the presumed direction between the nodes of Alice is considered, this opens up the possibility of violating causal inequalities. This shows that the standard PM framework, applied naively on the nodes over which the parties can operate, does not suffice to capture the causal nonseparability of the process on those nodes: the in-principle possibility of using side channels, which underlies the MPM concept, can have nontrivial consequences.
\section{\label{sec:theory}The framework}
Consider a set of parties, $\mathfrak{N}\equiv \left\{A,B,C,\ldots\right\}$, of finite cardinality $\left|\mathfrak{N}\right|$. In the original process matrix (PM) framework \cite{OCB2012}, each party $X$ is assumed associated with a pair of finite-dimensional input and output quantum systems, $X_{in}$ and $X_{out}$ with respective Hilbert spaces $\mathcal{H}^{X_{in}}$ and $\mathcal{H}^{X_{out}}$, hereby referred to as a \textit{(quantum) node} \cite{Allen2017,barrett2019quantum,barrett2020cyclic}. The party $X$ can perform an arbitrary causal quantum operation (also called \textit{quantum instrument}) from the input to the output. This is described by a collection of completely positive (CP) maps \cite{Nielsen2009} corresponding to the different possible outcomes $i_X$ of the operation, $\left\{ \mathcal{M}_{i_X}^{X_{in}\rightarrow X_{out} }\right\}$, $\mathcal{M}^X_{i_X}: \mathcal{L}(\mathcal{H}^{X_{in}})\rightarrow \mathcal{L}(\mathcal{H}^{X_{out}})$, where $\mathcal{L}(\mathcal{H})$ denotes the space of linear operators over $\mathcal{H}$. These CP maps obey the constraint that $\mathcal{M}^{ X_{in}\rightarrow X_{out} } \equiv \sum_{i_X}\mathcal{M}_{i_X}^{X_{in}\rightarrow X_{out} } $ is a completely positive and trace-preserving (CPTP) map. 

In the PM formalism, CP maps are represented via positive semidefinite operators using the Choi-Jamio{\l}kowski (CJ) isomorphism \cite{Jamiolkowski1972,Choi1975}. Here, we take a basis-independent version of the isomorphism following \cite{Allen2017,barrett2019quantum,barrett2020cyclic} (mudulo an overall transposition) for which a CP map $\mathcal{M}^{X_{in}\rightarrow X_{out} }$ is represented by a CJ operator $M^{X_{in}X^*_{out}} \in \mathcal{L}\left(\mathcal{H}^{X_{in}} \otimes (\mathcal{H}^{X_{out}})^*\right)$, where $(\mathcal{H}^{X_{out}})^*$ is a copy of the dual of $\mathcal{H}^{X_{out}}$ \footnote{Note that this copy of the dual of the output Hilbert space can be seen related to the time-reversed version of the output Hilbert space \cite{OreshkovCerf2016} via a concrete linear map, which was incorporated in the time-neutral generalisation of the formalism developed in Ref. \cite{OreshkovCerf2016}. We do not invoke this here for simplicity.}\nocite{OreshkovCerf2016}, via %
\begin{multline}\label{eq:CJ}
	M^{X_{in}X_{out}^*} \equiv \\ 
	\left(\sum_{i,j} \ \mathcal{M}^{X_{in}\rightarrow X_{out} }(\ket{i}\bra{j}^{X_{in}}) \otimes \ket{i}\bra{j}^{X^*_{in} }\right)^{\textrm{T}}\,,
\end{multline}
where $\left\{ \ket{i}^{X_{in}} \right\}$ is an orthonormal basis of $\mathcal{H}^{X_{in}}$, $\left\{ \ket{i}^{X^*_{in}} \right\}$ the corresponding dual basis, and we identify $(\mathcal{H}^*)^*$ with $\mathcal{H}$ via the canonical isomorphism. For a CPTP map $\mathcal{M}^{X_{in}\rightarrow X_{out}}$, $\Trace_{X_{out}^*}[M^{X_{in}X_{out}^*}] = \mathds{1}^{X_{in}}$. 

Given a choice of instrument at each node, the joint probabilities for the outcomes of all parties are then given by the `generalised Born rule' 
\begin{gather}
p(i_A,i_B, \ldots| \{\mathcal{M}^{A_{in}\rightarrow A_{out}}_{i_A}\}, \{\mathcal{M}^{B_{in}\rightarrow B_{out}}_{i_B}\}, \ldots) = \notag\\
 \textrm{Tr}\left[W^{A_{in}A^*_{out}B_{in}B^*_{out}\ldots} \left(M^{A_{in}A^*_{out}}_{i_A}\otimes M^{B_{in}B^*_{out}}_{i_B}\ldots \right)\right],\label{eq:PMBorn}
\end{gather}
where $W^{A_{in}A^*_{out}B_{in}B^*_{out}\ldots}\in \mathcal{L}( \mathcal{H}^{A_{in}}\otimes (\mathcal{H}^{A_{out}})^* \otimes \mathcal{H}^{B_{in}}\otimes (\mathcal{H}^{B_{out}})^* \ldots)$. This follows simply from the assumption that, as in standard quantum theory, the probabilities are linear functions of the CP maps corresponding to the outcomes. The operator $W^{A_{in}A^*_{out}B_{in}B^*_{out}\ldots}$ is called the \textit{process matrix (PM)} (also \textit{process operator} \cite{barrett2019quantum, barrett2020cyclic}). The only conditions it has to satisfy are $W^{A_{in}A^*_{out}B_{in}B^*_{out}\ldots}\geq 0$ (coming from the requirement of non-negativity of the probabilities, assuming the parties' operations can be extended to act on local input ancillas prepared in arbitrary joint quantum states), and $\textrm{Tr}\left[W^{A_{in}A^*_{out}B_{in}B^*_{out}\ldots} \left(M^{A_{in}A^*_{out}}\otimes M^{B_{in}B^*_{out}}\ldots \right)\right]=1$ on all CPTP maps $M^{A_{in}A^*_{out}}, M^{B_{in}B^*_{out}}, \ldots$ (coming from the requirement that probabilities sum up to 1). Practical necessary and sufficient conditions for an operator to be a valid PM have been formulated based on the types of nonzero terms appearing in the expansion of the operator in a Hilbert-Schmidt basis \cite{OCB2012, OG2016}, as well as based on a superoperator projector \cite{Witness}. In the next section, we will use the latter, of which we will derive a generalisation. A review of it can be found in Appendix \ref{app:projo}.
\begin{figure}
 \centering
 \includegraphics[height=.9\linewidth]{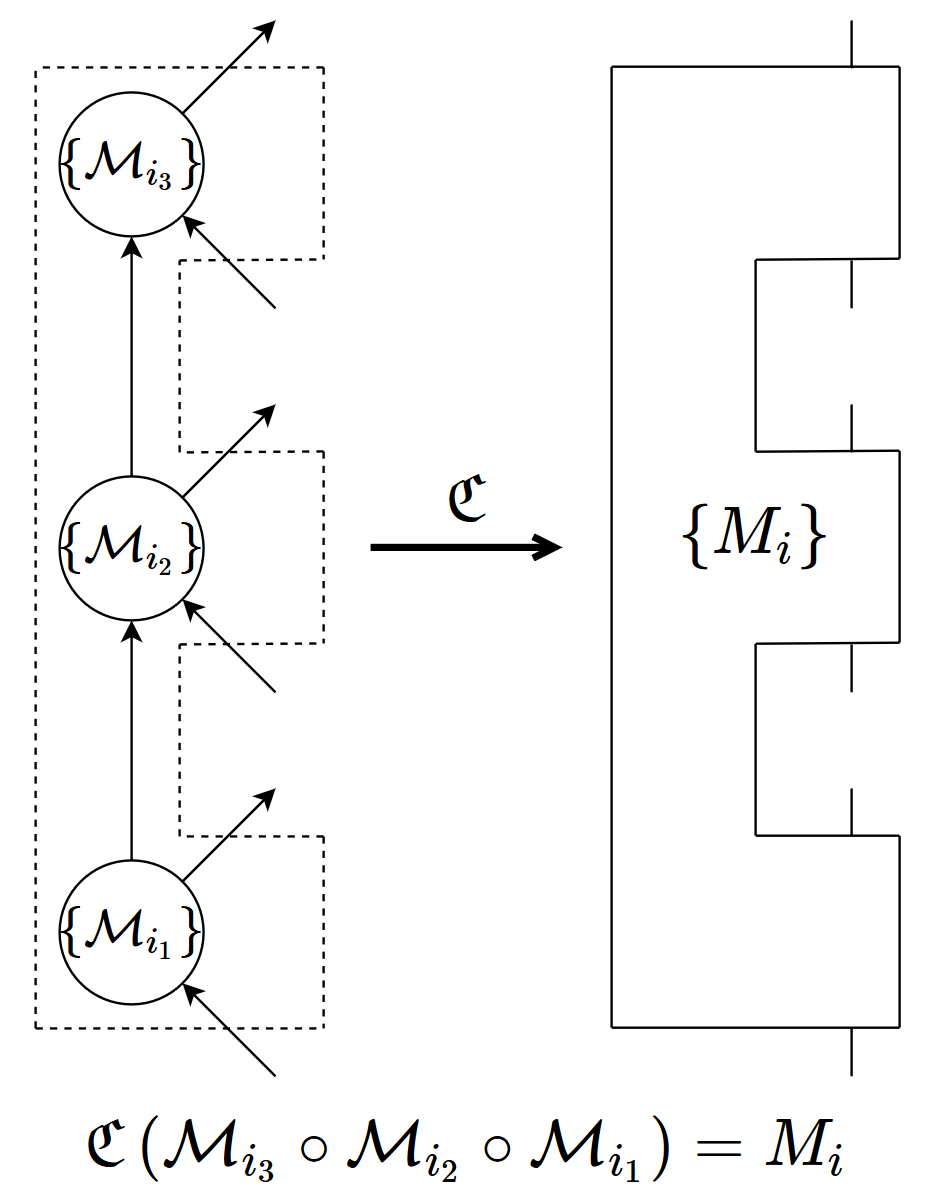}
 \caption{Graphical representation of a quantum network with 3 nodes (plain circles), and its associated quantum 3-instrument obtained through CJ isomorphism $\mathfrak{C}$ applied on the composition of the operations at each node.}
 \label{fig:QuComb}
\end{figure}

Formula \eqref{eq:PMBorn} can be interpreted as the result of composing in a loop the local operations of the parties with a channel $\tilde{\mathcal{W}}$ from the output systems of all nodes to the inputs systems of all nodes: the PM is the \textit{transpose} of the CJ operator of $\tilde{\mathcal{W}}$, where the transposition simply reflects the \textit{link product} \cite{Chiribella2009} for composing channels in the CJ representation. 

Here, we generalise the framework by relaxing the assumption that each party is restricted to a single round of receiving a system in and sending a system out. We now assume that each party is associated with an ordered sequence of quantum nodes over which they can operate in a causal fashion. Let $n_X$ denote the number of nodes for party $X$. There are then $2n_X$ systems associated with $X$. We will label them by $X_{j}$, $j=0,\ldots 2n_X-1$, with $X_0$ and $X_1$ being respectively the input system and the \textit{dual} of the output system of the first node, $X_2$ and $X_3$ the input system and the dual of the output system of the second node, and so on, with even (respectively, odd) numbering referring to an input (dual of output) system. The $i$-th node of $X$ will be compactly denoted by $X^{(i)}$. In other words, for $i= 1,\ldots, n_X$, with $X_{i,in}$ (respectively, $X_{i,out}$) referring to the input (output) of the $i$-th node we define
	\begin{subequations}
	\begin{align}
	&\Hilb{X_{2i-2}} \equiv \Hilb{X_{i,in}} \quad, \\
	&\Hilb{X_{2i-1}} \equiv \left(\Hilb{X_{i,out}}\right)^* \quad, \\
	&\Hilb{X^{(i)}} \equiv \Hilb{X_{i,in}} \otimes \left( \Hilb{X_{i,out}}\right)^* \quad.
	\end{align}
	\end{subequations}

The most general causal quantum operation on a sequence of $n$ nodes is a non-deterministic quantum network \cite{Chiribella2009}, which can be implemented as a sequence of $n_X$ instruments acting on the given nodes plus ancillary systems connected via side channels, as illustrated on the left-hand side of Fig. \ref{fig:QuComb} for the case of 3 nodes. Such a network amounts to implementing a quantum operation from the joint input system of all nodes to the joint output system of all nodes, called an \textit{n-instrument} \cite{Chiribella2009}. In the CJ representation, an $n_X$-instrument is described by a collection of positive semidefinite (PSD) operators $\{M^{X_0X_1\cdots X_{2n_X-1} }_i\}$, defined on $\mathcal{L}(\mathcal{H}^{X_0}\otimes \mathcal{H}^{X_1} \ldots\otimes \mathcal{H}^{X_{2n_X-1}})$. We will shorten the notation using $X_0X_1\cdots X_{2n_X-1} \equiv X$  to lessen clutter when there is no ambiguity, writing $\left\{M^X_i\right\} \in \LinOp{\Hilb{X}}$. Each $M^{X}_i$ is labelled by the outcome index $i$, which in this case is a poly-index corresponding to the collection of outcomes at the different steps in the network, $i\equiv (i_{\Party{X}{1}{}}, \ldots, i_{\Party{X}{n_X}{}})$. The operators $M^{X}_i$ are the CJ operators of CP maps from the joint input system of all nodes to the joint output system of all nodes, called \textit{probabilistic quantum ($n_X$-)combs} \cite{Chiribella2009}. These operators must satisfy the condition that $M^X \equiv \sum_i M^X_i$ is a \textit{deterministic quantum ($n_X$-)comb} \cite{Chiribella2009}. The latter is the CJ operator of a deterministic $n_X$-instrument, \textit{i.e.} the most general CPTP map from the joint input system of the nodes to the joint output system of the nodes that can be implemented via a network of channels as in the left-hand side of Fig. \ref{fig:QuComb}. An operator $M^X$ is a deterministic quantum $n_X$-comb if and only if it obeys the following constraints \cite{Chiribella2009}:%
\begin{subequations}\label{eq:QuComb_Chiri}
    \begin{gather}
    M^{X_0\cdots X_{2n_X-1}}\geq 0 \,, \label{eq:QuComb_Chiri_pos}\\
    \begin{split}
    \forall j\, : \, 2\leq j \leq n_X ,&\\
    \TrX{X_{2j-1}}{M^{X_0\cdots X_{2j-1}}} &=
    \mathds{1}^{X_{2j-2}} \otimes M^{X_0 \cdots X_{2j-3}},\\
    \TrX{X_1}{M^{X_0X_1}} &= \mathds{1}^{X_0}\,,
    \end{split}\label{eq:QuComb_Chiri_subspace} %
    \end{gather}
\end{subequations}
where $\mathds{1}^Y$ denotes the unit matrix on subsystem $Y$.

Analogously to the PM framework, we assume that the joint probabilities for the outcomes of the generalised instruments performed by the separate parties in this multi-round setting are linear functions of the probabilistic quantum combs associated with the outcomes, as would be the case in standard quantum mechanics if the nodes are associated with quantum systems at definite times. This means that we can write these probabilities in the `generalised Born rule' form
\begin{gather}
p(i_A,i_B, \ldots| \{ M^{A} _{i_A}\}, \{{M}^{B}_{i_B}\}, \ldots) \notag\\
= \textrm{Tr}\left[W^{AB\ldots}\cdot \left(M^{A}_{i_A}\otimes M^{B}_{i_B}\ldots \right)\right],\label{eq:MPMBorn}
\end{gather}
where for the superscripts that label the systems over which the operators are defined we have used the short-hand notation $X\equiv X_0X_1\ldots X_{2n_X-1}$, $X\in\mathfrak{N}$. We call the operator $W^{AB\ldots}$ the \textit{multi-round process matrix (MPM)} (see Fig. \ref{fig:MPM_example}) . 

\begin{figure}
\centering
\includegraphics[width=.9\linewidth]{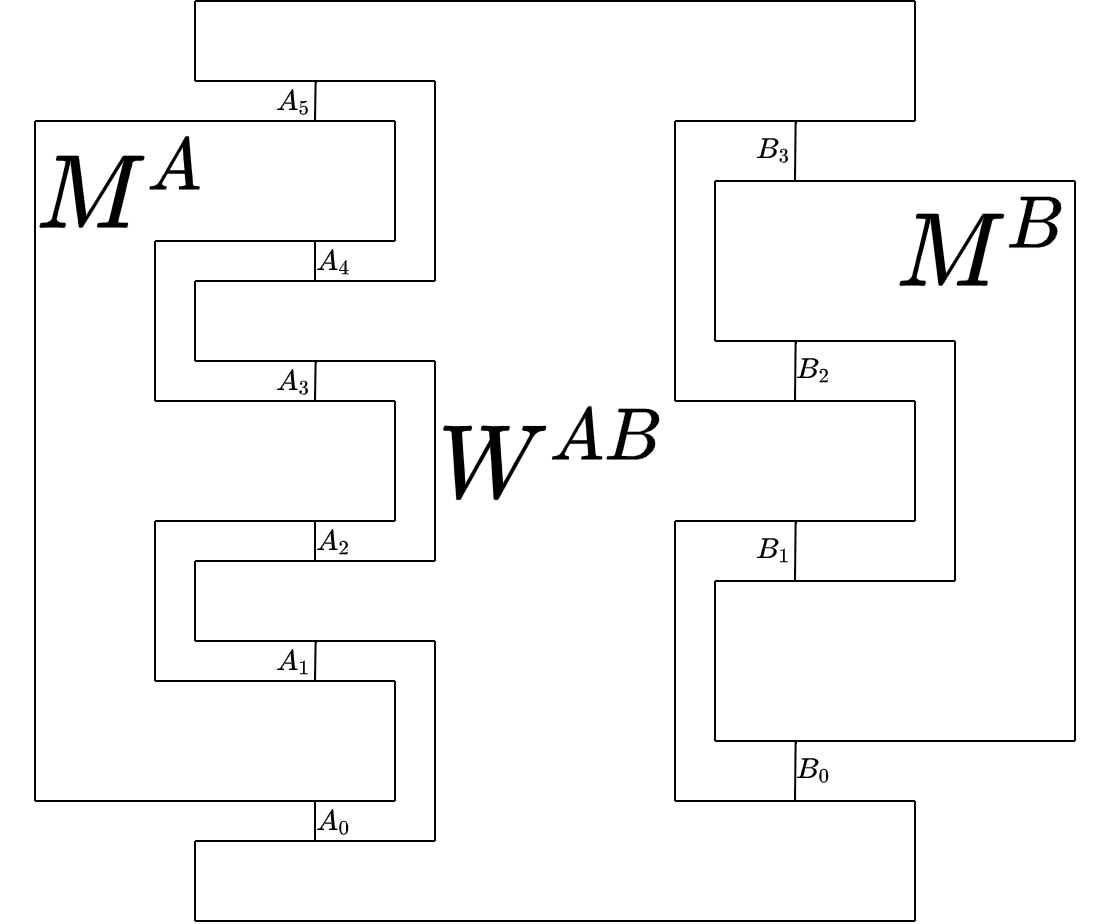}
\caption{Graphical representation of an MPM for 2 parties A and B. A has 3 rounds and B has 2. Here, the name of the subsystems associated with each wire is apparent. By convention, the ordering of the indices goes from bottom to top}
\label{fig:MPM_example}
\end{figure}

As in the PM framework, the only property we demand from $W^{AB\ldots}$ is that it yields valid probabilities through Eq. \eqref{eq:MPMBorn} for all generalised instruments that can be applied by the parties, including when the instruments are extended to act on ancillary input systems in arbitrary quantum states. It is straightforward to see in analogy to the argument in Ref. \cite{OCB2012} that this is equivalent to the following constraints:
\begin{gather}
W^{AB\ldots} \geq 0 \, ,\label{eq:MPM_pos}
\end{gather}
and
\begin{gather}
\textrm{Tr}\left[W^{AB\ldots}\cdot \left(M^{A}\otimes M^{B}\ldots \right)\right] = 1 \label{eq:MPM_norm}
\end{gather}
for all deterministic quantum combs $M^{A}, M^{B}, \ldots$.

\begin{figure}
 \centering
 \includegraphics[height=.9\linewidth]{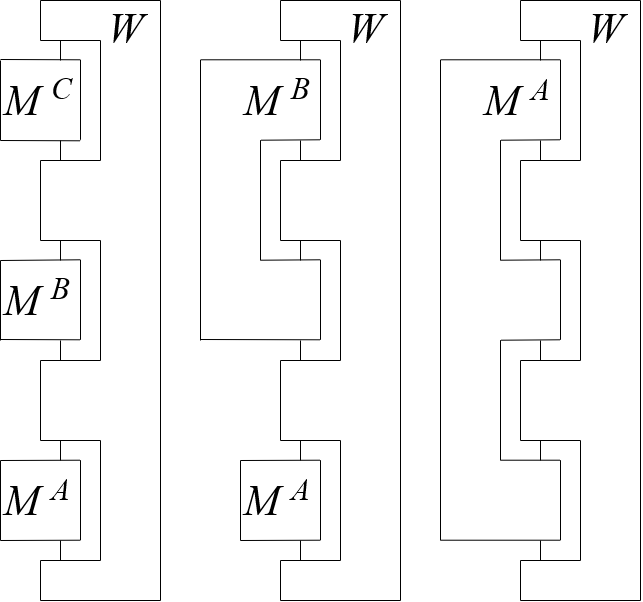}
 \caption{Graphical representations of an MPM $W$ for: (left) 3 parties with a single node each [$W$ is a PM in this case]; (center) a party $A$ having a single round and $B$ having two rounds; (right) a single party $A$ with 3 rounds [$W$ is a deterministic quantum comb in this case].}
 \label{fig:PMvMPMvComb}
\end{figure}

\textit{Remark 1.} In the case when the number of nodes per party is 1, $W^{AB\ldots}$ reduces to a standard process matrix, whereas when the number of parties is 1 and that party has $n$ nodes, $W^{A}$ reduces to (the transpose of) a deterministic quantum $(n+1)$-comb with a trivial first input system and a trivial last output system \cite{Chiribella2009}. See Fig. \ref{fig:PMvMPMvComb}.


\textit{Remark 2.} Note that if an operator $W^{AB\ldots}$ is a valid MPM, it is in particular a valid PM if each different node is interpreted as belonging to a separate party. This is because the probabilities must be equal to 1 when the parties perform independent CPTP maps in their different nodes, which is a special case of a deterministic network. The converse, however, is obviously not true since a PM is not normalised on deterministic quantum combs on arbitrary ordered subsets of its nodes (\textit{e.g.} if the PM is equivalent to a deterministic quantum comb for a specific order of the nodes, it would generally not admit compositions with deterministic quantum combs in the opposite order).

\section{Characterisation of the set of valid MPMs\label{sec:charactMPM}}
We now derive necessary and sufficient conditions for an operator to be a valid MPM that are easier to handle. They can be seen as a generalisation of the projective characterisation of the set of PMs obtained in Ref. \cite{Witness}, where the authors showed that the set of valid PMs belongs to a subspace of the space of operators, the projector on which can be deduced from the one on valid quantum channels. 

Our approach goes as follows. Given the projective characterisation of the space of valid quantum channels in the CJ picture \cite{Witness}, we can iteratively infer a similar projective characterisation of the space of deterministic quantum combs. This is done in Appendix \ref{app:CombCd}, and we find that all the deterministic quantum combs $M^X$, no matter their number of nodes (or `teeth'), have the same algebraic structure: 
\begin{subequations}\label{eq:QuComb}
    \begin{align}
    &M^X \geq 0 \quad,\\
    &\ProjOn{X}{n_X}{M^X}=M^X \quad,\\
    &\TrX{}{M^X} = \prod_{j=0}^{n-1} d_{X_{2j}} \equiv  d_{X_{in}} \,. \label{eq:QuComb_norm}
    \end{align}
\end{subequations}
That is, they are elements of a positive and trace-normalised subset within a subspace of the space of operators on the full system. The subspace is defined through a projector $\Proj{X}{n_X}$, whose superscript refers to the system it acts upon (we will refine the signification of the superscript in the next section) and whose subscript indicates the total number of nodes. As a function of $X$, the exact form of the projector is given by the recursive relation 
\begin{equation}\label{eq:ProjComb}
        \mathcal{P}_{n_X}^{X} =\mathcal{I}^X - \DepPar{X_{2n_X-1}}{\cdot} + \DepPar{X_{2n_X-2}X_{2n_X-1}}{\Proj{X'}{n_X-1}} \, ,
\end{equation}
expressed in terms of the mapping $\DepPar{X_i}{\cdot} \equiv \MapX{X_i}{\cdot}$ introduced in Ref. \cite{Witness} (see Appendix \ref{app:projo}) and the identity projector $\mathcal{I}^X\left[O \right] = O \,, \forall \; O \in \LinOp{\Hilb{X}}$. $\Proj{X'}{n_X-1}$ is the $(n_X-1)$-comb projector that acts on subsystems $X' \equiv X_0\ldots X_{2n_X-3}$ (implicitly understood extended by a tensor product with the identity on $X_{2n_X-2}X_{2n_X-1}$). The recursion starts with the 0-comb projector being 1, $\Proj{}{0} \equiv 1$, because the space of 0-combs is the real numbers.

Here, we are interested in characterising the set $\{W\}$ of valid MPMs as defined in Eqs. \eqref{eq:MPM_pos} and \eqref{eq:MPM_norm}. Note that condition \eqref{eq:MPM_norm}, which states that an MPM is normalised on all tensor products of deterministic quantum combs, is equivalent to the condition that an MPM is normalised on all affine combinations of tensor product of such combs:%
\begin{equation}\label{eq:MPM_norm_affine}
    \TrX{}{W^{AB\ldots}\cdot \sum_i q_i \left(M^{A}\otimes M^{B}\ldots \right)_i}= 1\,, 
\end{equation}
with $\sum_i q_i = 1$ while $\left( M^A \otimes M^B \ldots  \right)_i$ is a tensor product of deterministic quantum combs that may be different for each index $i$. The fact that \eqref{eq:MPM_norm_affine} is necessary for \eqref{eq:MPM_norm} follows from the linearity of the trace, while its sufficiency follows from the fact that \eqref{eq:MPM_norm} is a special case of \eqref{eq:MPM_norm_affine}.

In Appendix \ref{app:Complement}, we show that an operator $M$ is an affine combination of tensor product of $|\mathfrak{N}|$ deterministic combs if and only if it respects the conditions
\begin{subequations}\label{eq:AffineCombs}
    \begin{align}
        &\ProjOn{}{}{M}\equiv \left( \bigotimes_{X\in \mathfrak{N}} \Proj{X}{n_X}\right)[M] = M \,, \label{eq:AffineCombs_proj} \\
        &\TrX{}{M} = \prod_{X \in \mathfrak{N}} d_{X_{in}} \equiv d_{in} \,, \label{eq:AffineCombs_norm}
    \end{align}
\end{subequations}
where we have defined the notations $\Proj{}{}$ for the overall projector and $d_{in}$ for the input dimension of the full Hilbert space. In the following, $\mathrm{Span}\left\{ \bigotimes M \right\}$ will be used to refer quickly to the subspace of operators spanned by a tensor product of deterministic combs, \textit{i.e.} operators satisfying \eqref{eq:AffineCombs_proj}. 

With these considerations, the rationale in Ref. \cite{Witness} of inferring PM validity conditions from the projector on the space of quantum channels and the normalisation can be generalised to MPMs using the projective characterisation \eqref{eq:AffineCombs} of the affine combination of tensor products of deterministic combs. This result is systematised as a theorem:%
\begin{theo}\label{theo:complement}
	Let $\LinOp{\Hilb{}}\equiv \LinOp{\Hilb{in} \otimes \left(\Hilb{out}\right)^*}$ be a linear space of operators defined on a Hilbert space of dimension $d\equiv d_{in}d_{out}$. This Hilbert space admits a tensor factorisation into $2|\mathfrak{N}|$ subsystems associated with the inputs and outputs of $|\mathfrak{N}|$ parties. Let $ \{ W \} \subset \LinOp{\Hilb{}}$ be the set of valid MPMs shared between the parties. Let $\Proj{X}{n_X}$ be the projector characterising the validity subspace of the deterministic combs that can be applied by party $X$. Let $\mathcal{P} \equiv \bigotimes_{X \in \mathfrak{N}} \Proj{X}{n_X}$ be the tensor product of the different such projectors of all parties. Then, an operator $W$ belongs to the set $ \{ W\}$ if and only if it satisfies
	\begin{subequations}\label{eq:MPM}
        \begin{align}
            &W\geq 0 \,, \label{eq:pos}\\
            &\CompProj{}{}[W] \equiv \left( \mathcal{I} - \mathcal{P} + \mathcal{D}\right)[W] = W\,, \label{eq:complProj}\\
            &\TrX{}{W} =  d_{out}\,, \label{eq:norm}
        \end{align}
    \end{subequations}
	where $\mathcal{I}$ is the identity projector, and $\mathcal{D}$ is the projector on the span of the unit matrix, $\mathcal{D}[O] = \MapX{}{O}$.%
\end{theo}%
The projector $\CompProj{}{}$ defined in Eq. \eqref{eq:complProj} will be referred to as the \textit{quasiorthogonal} projector to $\Proj{}{}$. The proof of this theorem is presented in Appendix \ref{app:Complement}. It can be understood in a more mathematical terminology as \enquote{the CJ representation of the generalised instruments is a subset of a unital subalgebra which is the quasiorthogonal complement to the subalgebra in which lies the set of valid processes $\{W\}$} \cite{Petz2007,hiai2014}. Equivalently, it says that the traceless part of an MPM is orthogonal (with respect to the Hilbert-Schmidt inner product) to the traceless part of the CPTP maps that are plugged into it \cite{Bisio2018}.

We note that an equivalent characterisation has been independently explored in \cite{Perinotti2016,Bisio2018}, but without using projective methods. In particular, Theorem \ref{theo:complement} corresponds to the Lemma 4 of \cite{Bisio2018}. Indeed, in Perionotti and Bisio's classification of higher-order quantum processes, the MPM is described by a \textit{type} $\left(\left(\mathbf{n}_A \otimes \mathbf{n}_B \otimes ... \right)\rightarrow 1\right)$. That is, a structure taking in a tensor product of quantum combs of types $\mathbf{n}_A, \mathbf{n}_B, \ldots$ and sending them onto the trivial type 1. Theorem \ref{theo:complement} is therefore a direct way of finding the constraints to apply on an operator so that it is the representation of such a type. In a subsequent work, we will present a systematic way of associating a projector to a type structure.

%
\section{Link between comb and MPM subspaces\label{sec:MPM=Combs}}
Using the above projective characterisation, we are able to highlight the link between the validity subspace of deterministic combs with trivial first input and last output and that of MPMs. 
We will need to introduce the symbol $\prec$ to mark the causal relation \enquote{is before}, indicating that a node is in the causal past of another. The projector on the subspace of valid 3-combs with teeth named $A,B,C$ so that $A\prec B \prec C$ will be denoted $\Proj{A \prec B \prec C}{3}$, with the total number of teeth indicated as a subscript and the causal ordering as a superscript. Permutations of the ordering will be denoted $\pi_i\left(A, B, C\right)$, with the subscript referring to a particular permutation, \textit{e.g.} $\pi_0\left(A, B, C\right)  = A \prec B \prec C\,, \pi_1\left(A, B, C\right)  = B \prec A \prec C,\ldots$ . When the letters refer to an ordered set of teeth, for example $A=\left\{\Party{A}{1}{}\right\}$ and $B=\left\{\Party{B}{1}{}\prec \Party{B}{2}{}\right\}$, the permutations are limited to those that are compatible with the the ordering between the elements within each set. So for this particular example, compatibility means that permutations containing a causal relation where $\Party{B}{2}{}$ appears before $\Party{B}{1}{}$ are forbidden, giving $\{\pi_i(A, B)\} = \{\Party{A}{1}{} \prec \Party{B}{1}{} \prec \Party{B}{2}{},\Party{B}{1}{} \prec \Party{A}{1}{} \prec \Party{B}{2}{}, \Party{B}{1}{} \prec \Party{B}{2}{} \prec \Party{A}{1}{}\}$.


Observe that Theorem \ref{theo:complement} states that the projector on the subspace of valid MPMs is quasiorthogonal to the projector on $\mathrm{Span}\left\{ \bigotimes M \right\}$, where deterministic combs in the tensor product are associated with the different parties. For a single party $A$, it is expected that the MPM is a deterministic comb \cite{Chiribella2009} (see Fig. \ref{fig:PMvMPMvComb}, rightmost case). Indeed, it can be verified using Eq. \eqref{eq:ProjComb} that the quasiorthogonal projector to a single comb projector is a comb projector itself; as expected, it sends to the validity subspace of deterministic $(n_A+1)$-combs with trivial first input and last output, and with node ordering compatible with the ordering of the teeth of $A$. By \enquote{compatible}, we mean here that the nodes in the $(n_A+1)$-comb, which are now the gaps between its teeth, are properly ordered, \textit{i.e.} $\Party{A}{1}{} \prec \Party{A}{2}{} \prec \ldots$ . Since we are working with orthogonal projectors, equivalence between projectors is sufficient to prove the equivalence between subspaces.

We can actually use the algebra of these projectors described in Appendix \ref{app:projo} to generalise this observation to any number of parties. For $|\mathfrak{N}|$ parties sharing an MPM with $n=n_A + n_B+\ldots$ nodes, what we find is that the MPM validity subspace is the span of all the deterministic $(n+1)$-combs with trivial first input and last output such that their node ordering is compatible with the local ordering of each of the $|\mathfrak{N}|$ combs that are to be plugged into it. 

For example, if Alice has one operation and Bob two, there are three possible such 4-combs compatible with $\Party{B}{1}{}\prec \Party{B}{2}{}$, having projectors $\CompProj{\Party{A}{1}{} \prec \Party{B}{1}{} \prec \Party{B}{2}{}}{3}$, $\CompProj{\Party{B}{1}{} \prec \Party{A}{1}{} \prec \Party{B}{2}{}}{3}$, and $\CompProj{\Party{B}{1}{} \prec \Party{B}{2}{} \prec \Party{A}{1}{}}{3}$ sending to three different subspaces. The validity subspace of the MPM is then spanned by the union of these three subspaces.

To make this claim explicit, we need to introduce the notion of \textit{union of projectors} \cite{Piziak1999}. For two arbitrary linear projectors $\mathcal{P}$ and $\Proj{'}{}$ acting on the same space, the projector on the span of the union of their subspaces is the union of the projectors given by
\begin{equation}
    \Proj{}{} \cup \Proj{'}{} = \Proj{}{} + \Proj{'}{} - \Proj{}{} \Proj{'}{} \, ,
    \label{eq:UprojCtilde}
\end{equation}
provided they commute. Thus, for $\mathfrak{M},\mathfrak{N} \subset \LinOp{\Hilb{}}$ such that $\mathfrak{M} = \left\{M | \ProjOn{}{}{M}=M\right\}$ and $\mathfrak{N} = \left\{N | \ProjOn{'}{}{N}=N\right\} $, we have $O \in \text{Span}\left\{ \mathfrak{M} \cup \mathfrak{N} \right\} \iff \left(\Proj{}{}\cup \Proj{'}{}\right)[O] = O$.

This allow us to formally state the result as a theorem:
\begin{theo} \label{theo:PV=UPTildeC}%
    For a set of parties $\{A,B,\ldots\}$ with each having possibly more than one node, let $\{\pi_i\}$ denote the set of valid permutations between the nodes as defined above.
    The projector on the linear subspace of valid MPMs shared between the parties is then equivalent to the union of all the quasiorthogonal projectors to the projectors on the linear subspaces of the deterministic quantum combs that respect the partial ordering of the teeth for each party:%
        \begin{equation}\label{eq:PV=UPTildeC}%
            \prescript{Q}{}{\left(\Proj{A}{n_A}\otimes\Proj{B}{n_B}\otimes\ldots\right)} = \bigcup_{\pi_i} \CompProj{\pi_i(A,B, \ldots)}{n_A+n_B+ \ldots} \,.
        \end{equation}
\end{theo}
A proof of this theorem is given in Appendix \ref{app:proof:V=UC}.  It shows that any MPM, and as a special case any PM, is a linear combination of deterministic quantum combs that respect the local ordering of each party. The trace condition \eqref{eq:norm} further constrains the coefficients in the combination to sum up to one, hence this is an affine combination.
\section{(Non)causal correlations and causal (non)separability for MPMs\label{sec:Res_CorrforMPM}}%
Now that the set of valid MPMs has been characterised, a natural ensuing question is which ones describe causal correlations between the nodes. On a purely theory-independent account, let the set of nodes $\Party{A}{1}{},\ldots,\Party{A}{n_A}{},\Party{B}{1}{},\ldots ,\Party{B}{n_B}{},\Party{C}{1}{},\ldots$ receive settings $\Vec{s}\equiv \left(s_{\Party{A}{1}{}},s_{\Party{A}{2}{}},\ldots \right)$ and produce outcomes $\Vec{o}\equiv \left(o_{\Party{A}{1}{}},o_{\Party{A}{2}{}},\ldots \right)$ after performing their operations. We can formulate a notion of the correlations $p(\Vec{o}|\Vec{s})$ being \textit{causal}, building upon the notion of causal correlations introduced for the original process framework in Ref. \cite{OG2016}. 

The idea of that definition is that the correlations are deemed casual if they are compatible with a stochastic unravelling of the events at the nodes in a causal sequence, which can be described as follows: first, according to some probability distribution, one of all nodes is selected to come first; then, depending on the setting chosen at that node, a specific outcome occurs at it with a specific probability; after that, one of the remaining nodes is selected to come second according to a probability distribution that may generally depend on the setting and outcome that have occurred at the first node; this continues until all nodes have been used, with the probability for each subsequent event generally depending on all events that have occurred in the past. In the multiround case, the same idea can be formalised as follows:
\begin{defi}[Causal multi-round correlations]
     We call the conditional probability distribution $p(\Vec{o}|\Vec{s})$ causal if and only if it can be decomposed as%
     \begin{equation}\label{eq:CnS:causal_corr}
        \begin{split}
        &p(\Vec{o}|\Vec{s}) \\
        &=\sum_{X\in \mathfrak{N}} q_{\Party{X}{1}{}} \; p(o_{\Party{X}{1}{}}|s_{\Party{X}{1}{}}) p(\Vec{o}_{\setminus\Party{X}{1}{}}|o_{\Party{X}{1}{}},\Vec{s}) \,,
        \end{split}
    \end{equation}
    with $q_{\Party{X}{1}{}} \in [0;1]$, $\sum_{X\in \mathfrak{N}} q_{\Party{X}{1}{}} = 1$, where $p(o_{\Party{X}{1}{}}|s_{\Party{X}{1}{}})$ is a single-node distribution and $p(\Vec{o}_{\setminus \Party{X}{1}{}}|o_{\Party{X}{1}{}},\Vec{s})$ is itself a conditional causal probability distribution on the remaining nodes after the relabelling $\Party{X}{i}{}\rightarrow\Party{X}{i-1}{}$ for $i=2,\ldots,n_X$ and $n_X \rightarrow n_X-1$. We have used the shorthand notation $\Vec{o}_{\setminus \Party{X}{1}{}} \equiv \Vec{o} \setminus \left\{o_{\Party{X}{1}{}}\right\}$ for the vector of remaining outcomes. 
\end{defi}
This definition is similar to the one for the single-round case \cite{OG2016,Abbott2016}, with the difference that only unravellings compatible with the local causal orders of the nodes of the parties are permitted.  This extra condition is a nontrivial addition: the fact that the overall correlations must be compatible with no signalling from future nodes to past nodes does not guarantee that if the correlations admit a causal unravelling in principle, that unravelling would respect the local orders. Consider for example the correlations between a party $A$ having two nodes and a party $B$ having one node. Let there be no signalling from $\Party{A}{2}{}$ to $\Party{A}{1}{}$. Hence, the correlations are in principle compatible with the assumed local causal ordering. However, with this assumption alone, there can be perfect signalling from $\Party{A}{2}{}$ to $B$ and from $B$ to $\Party{A}{1}{}$, as long as $B$ does not forward any information obtained from $\Party{A}{2}{}$ to $\Party{A}{1}{}$. In such a scenario, no unravelling is compatible with the supplementary constraint that $\Party{A}{1}{}$ is before $\Party{A}{2}{}$. 

The new definition allows us to speak about \textit{causal inequalities} for the MPM, which are bounds on the set of joint probability distributions whose violation by a specific distribution implies that the distribution is not causal \cite{OCB2012}. Investigating the subject of causal inequalities for MPMs is left for future work.

In the process formalism, a given PM is \textit{causally separable} if and only if it produces causal correlations in which the conditional probability distributions appearing in the causal unravelling can themselves be seen as arising from a quantum process \cite{OG2016} (see precise definition later in the case of the MPM). 
Note that while causal separability implies causal correlations, the converse is not true. In accordance with this definition, we want to define causal separability for MPMs as the ability to unravel the causal ordering of the operations. This unravelling should be compatible with the partial ordering existing between the nodes associated to each party, and it should also take into account the effects of the new means of communication provided by the side channels between those nodes. As it turns out, this requirement motivates non-trivial modifications to the current definition for PMs \cite{OG2016,Wechs2018}. Noticing that every generalised instrument performed by a given party can be thought of as consisting of independent single-node operations connected via side channels (see below), a concise way of taking this into account is to apply the PM definition of causal separability on an MPM extended by identity side channels.
\begin{defi}\label{def:MPM_CnS}
    An MPM $W$ is causally separable if and only if the operator obtained by extending the nodes of each party with identity side channels between consecutive nodes is causally separable in the sense of a PM. In equation, let there be an extension of all the parties' subsystems in an MPM $W$ such that $ \bigotimes_{i=0}^{2n_X-1} \Hilb{X_i} \rightarrow  \bigotimes_{i=0}^{2n_X-1} \Hilb{X_i} \otimes \bigotimes_{j=1}^{2n_X-2} \Hilb{X'_{j}}$, where $d_{X'_{2k-1}} = d_{X'_{2k}}$, $k= 1 ,\ldots, n_X-1$, so that a node $\Party{X}{i}{}$, $1<i<n_X-1$, is now defined on $\Hilb{X_{2i-2}} \otimes \Hilb{X_{2i-1}} \otimes \Hilb{X'_{2i-2}} \otimes \Hilb{X'_{2i-1}}$, $\Party{X}{1}{}$ is defined on $\Hilb{X_0} \otimes \Hilb{X_1}\otimes \Hilb{X_1'}$, and $\Party{X}{n_X-1}{}$ is defined on $\Hilb{X_{2n_X-2}} \otimes \Hilb{X_{2n_X-1}}\otimes \Hilb{X'_{2n_X-2}}$. We define a PM $W'$ on the extended nodes as
    \begin{equation}\label{eq:MPMintoPM}
        W' = W \otimes \left( \bigotimes_{X\in\mathfrak{N}}\bigotimes_{i=1}^{n_X-1} Id^{X'_{2i-1}X'_{2i}} \right)\,,
    \end{equation}
    where $Id^{X'_{2i-1}X'_{2i}}$ is the transpose of the CJ operator of an identity side channel from $\left(\Hilb{X'_{2i-1}}\right)^*$ to $\Hilb{X'_{2i}}$. Then, $W$ is a causally separable MPM if and only if $W'$ is a causally separable PM for all dimensions of the identity side channels.
\end{defi}
For completeness, we now show explicitly that the correlations obtained through the `Born rule' \eqref{eq:MPMBorn} for MPMs can be transformed into correlations obtained with a PM extended by identity side channels. We will make use of a version of the \textit{link product} \cite{Chiribella2009}, which allows to merge combs together. In the context of our basis-independent convention for the CJ isomorphism, we define the link product between two operators $M^{AB^*}$ and $N^{BC^*}$, where $A$ and $C^*$ are separate systems or the trivial system, and $B$ and $B^*$ are mutually dual (\textit{i.e.} $\mathcal{H}^{B^*}$ and $\mathcal{H}^{B}$ are Hilbert spaces dual to each other, assuming the canonical isomorphism $(\mathcal{H}^*)^* \cong \mathcal{H}$ in the finite-dimensional case), as%
\begin{equation}\label{linkproduct}%
    \begin{gathered}%
         L^{AC^*} = M^{AB^*} \ast N^{B C^*}    : = \\
        \TrX{B B^*}{\left(M^{AB^*} \otimes N^{B C^*}\right) \left(\mathds{1}^{A} \otimes Id^{BB^*} \otimes \mathds{1}^{C^*} \right)},
    \end{gathered}%
\end{equation}
where $Id^{BB^*} \equiv \sum_{ij}^{d_B}  \dyad{ii}{jj}^{BB^*}$. 

\begin{figure}
\centering
\includegraphics[width=.9\linewidth]{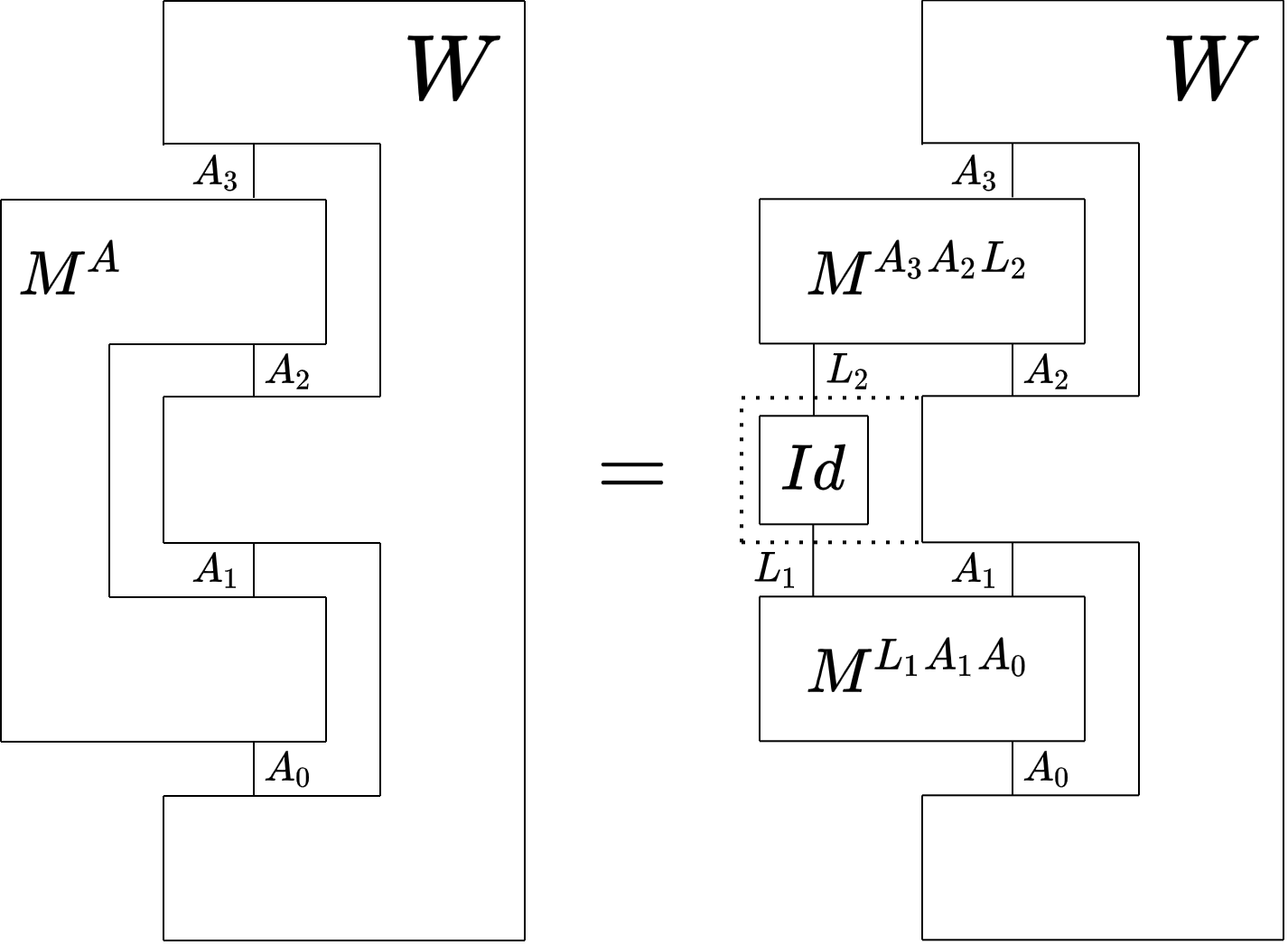}
\caption{Graphical interpretation of Eq.\eqref{eq:app:split}: an MPM (left) taking in a 2-comb is interpreted as a PM (right) extended by a side-channel (dashed part) taking in two 1-combs.}
\label{fig:MPM2PM}
\end{figure}
Consider a quantum 2-comb $M^A\equiv M^{A_0A_1A_2A_3}$. The CP map it describes can be seen as resulting from the composition of two CP maps with an intermediate identity side channel. In the CJ representation, this is given by the link product between the corresponding 1-combs,
\begin{equation}
    M^A = M^{A_3A_2L_2} \ast Id^{L_2^*L_1^*} \ast M^{L_1A_1A_0} \quad , \label{eq:app:M^A}
\end{equation}
where $L_1^*$ is the ancillary output system of the first CP map (with CJ operator $M^{L_1A_1A_0}$), $L_2$ is the ancillary input system of the second CP map (with CJ operator $M^{A_3A_2L_2}$), and $Id^{L_2^*L_1^*}$ is the CJ operator of the identity side channel connecting these systems. This is explicitly given by
\begin{multline}%
         M^A = \\ 
         \text{Tr}_{L_2L_2^*L_1L_1^*}\Big[ \left(M^{A_3A_2L_2} \otimes Id^{L_2^*L_1^*} \otimes  M^{L_1 A_1A_0}\right)  \\
         \times \left(\mathds{1}^{A_3A_2} \otimes  Id^{L_2L_2^*} \otimes Id^{L_1^*L_1} \otimes \mathds{1}^{A_1A_0}\right)\Big]\quad.
\end{multline}
Computing the trace over $L_2^*L_1^*$ yields
\begin{multline}\label{eq:app:LinkalaOO}
    M^A = \\
    \TrX{L_2 L_1}{\left(M^{A_3A_2L_2} \otimes M^{L_1 A_1A_0}\right) \left(\mathds{1}^{A} \otimes Id^{L_2L_1}\right)}.
\end{multline}
In this last expression, $Id^{L_2L_1} \equiv \sum_{ij}^{d_L^2} \dyad{ii}{jj}^{L_2L_1}$ is the \textit{transpose} of the CJ operator of an identity channel from $\Hilb{L_1^*}$ to $\Hilb{L_2}$ (just like a PM is the transpose of the CJ operator of a channel from the outputs of nodes to the inputs of nodes).  

Plugging equation \eqref{eq:app:LinkalaOO} into the `Born rule' between the 2-comb $M^A$ and a 2-node MPM $W$, one has 
\begin{multline}\label{eq:app:split}
    \TrX{}{M^A W} =\\
     \TrX{}{\left( M^{A_3A_2L_2} \otimes M^{L_1A_1A_0} \right) \left( W \otimes Id^{L_2L_1} \right)} ,
\end{multline}
where $ M^{A_3A_2L_2}, M^{L_1A_1A_0} $ are quantum 1-combs, and $ W \otimes Id^{L_2L_1} $ can be shown to be a valid (M)PM, hinting the formula \eqref{eq:MPMintoPM} for the general case. Indeed, formula \eqref{eq:app:M^A} can be applied repetitively in the case of a party with more than 2 nodes: $M^A = M^{{A}_{2n_A-2}{A}_{2n_A-1}{L}_{2n_A-2}} \ast\ldots    \ast M^{L_3A_3A_2L_2} \ast Id^{L_2^*L_1^*}\ast M^{{L}_{1}{A}_{1}{A}_{0}}$, and then the identity $\TrX{Y}{A^Y \cdot \TrX{X}{B^{XY}}} = \TrX{XY}{\left(\mathds{1}^X\otimes A^Y \right) B^{XY}}$ can be used to rewrite the whole expression as one overall partial trace. The tensor product being a special case of link product over the trivial system \cite{Chiribella2009}, this argument also applies to the multipartite case.

Consequently, the correlations achievable with an MPM are equivalent to those in a PM extended by side channels and we can appeal to the concept of causal separability for PMs in order to define causal separability for MPMs as it is done in Definition \ref{def:MPM_CnS}.

As shown in Ref. \cite{OG2016}, for the PM definition of causal separability, the possibility of extending the local operations of the parties to act on local ancillary input systems that may be entangled with the ancillary input systems of other parties also needs to be considered. Neglecting to do so may mistakenly lead to the conclusion that a PM is causally separable, although it is able to violate a causal inequality when using such an extension. Likewise, we will show in the next section that neglecting the side channels leads to the same kind of problem. We end this section by reformulating Definition \ref{def:MPM_CnS} into one closer to Ref. \cite{OG2016,Wechs2018}, so that the consequences of local causal ordering are made explicit and extension by arbitrary ancillas and side channels is explicit.

We only need to rewrite \eqref{eq:app:split} as
\begin{equation}
    \TrX{}{M^A W} = \TrX{A_3A_2L_2}{M^{A_3A_2L_2} W_{|M^{L_1A_1A_0}}}\,,
\end{equation}
where we have defined the \textit{conditional process matrix} 
%
%
\begin{multline}
    W_{|M^{L_1A_1A_0}} = \\
    c \;\TrX{L_1A_1A_0}{\left(  \mathds{1} 
     \otimes M^{L_1A_1A_0} \right) (W\otimes Id^{L_2L_1})} \,.
\end{multline}
This can be shown to be a valid process matrix on the remaining node between $A_2L_2$ and $A_3$ for all the possible choices of $M^{L_1A_1A_0}$, as long as the normalisation factor is chosen such that Eq. \eqref{eq:norm} is satisfied. On the contrary, the analogously defined $W_{|M^{A_3A_2L_2}}$ will not always be a valid PM for all $M^{A_3A_2L_2}$ as it can lead to post-selection. This implies that when we reformulate the `Born rule' so as to make the unravelling from an $n$-node MPM to an $(n-1)$-node MPM apparent, starting the unravelling with a node that is not the first of some party is automatically forbidden. The generalisation of this unravelling procedure yields the following definition:
\begin{defi}[Causal separability of the MPM]\label{def:MPM_CS}
    Consider an MPM $W$ shared by $|\mathfrak{N}|$ parties. For $|\mathfrak{N}|=1$ this MPM is a deterministic quantum comb and thus causally separable. For $|\mathfrak{N}|>1$, the MPM is causally separable if and only if, for any state $\rho \in \LinOp{\bigotimes_{X\in\mathfrak{N}}\bigotimes_{i=0}^{n_X-1} \Hilb{X_{2i}'}}$ defined on an extension of the parties' input subsystems $\Hilb{X_{2i}} \rightarrow \Hilb{\tilde{X}_{2i}} = \Hilb{X_{2i}} \otimes \Hilb{X_{2i}'}$, the extended MPM can be decomposed as
    \begin{equation}
        W \otimes \rho = \sum_{X\in\mathfrak{N}} q_{\Party{X}{1}{}} W^{\rho}_{\Party{X}{1}{}}\,,
    \end{equation}
    with $q_{\Party{X}{1}{}}\geq0$, $\sum_{X\in\mathfrak{N}} q_{\Party{X}{1}{}} = 1$, and where $W^{\rho}_{\Party{X}{1}{}}$ is an MPM compatible with party $X$'s first operation being first in the causal unravelling (\textit{i.e.} there can be no signalling from the rest of the nodes to $\Party{X}{1}{}$ \cite{OG2016}), so that the conditional MPM after the first operation of $X$ has been carried out,
    \begin{widetext}
    \begin{equation}
        \left( W^{\rho}_{\Party{X}{1}{}}\right)_{\left|M^{X_0X'_0X_1L_1}\right.} = c \; \TrX{X_0X'_0X_1L_1}{\left( M^{X_0X'_0X_1L_1} \otimes \mathds{1} \right)\left( W^X \otimes \rho^{X'} \otimes Id^{L_2L_1} \right)} \,,
    \end{equation}
    \end{widetext}
    is itself causally separable for all possible CP maps between $\LinOp{\Hilb{X_0}\otimes \Hilb{X'_0}}$ and $\LinOp{\left(\Hilb{X_1}\right)^* \otimes \left(\Hilb{L_1}\right)^*}$ represented by the CJ operator $M^{X_0X'_0X_1L_1}$. Here $L_2$ is an extension of the input system of $\Party{X}{2}{}$ and $c$ is a normalisation constant. 
\end{defi}

In the two equivalent definitions of causal separability above, we have left the dimensions of the side channels unbounded. Yet, it is natural to ask whether a bounded dimension could suffice. It is clear that no bounded dimension could reproduce all correlations achievable with an MPM because, for example, a given party could follow a protocol where she applies a different instrument at her second node depending on the outcome of the instrument applied at her first node. By taking the number of outcomes at the first node sufficiently large, the side channel required to realise this situation could be made larger than any assumed bound. However, it may still be the case in principle that some bounded dimension for the side channels (generally dependent on the dimensions of the nodes of the MPM) is sufficient for the definition of causal separability because it implies the analogous property for all dimensions. The question of whether such a bound exists is left open for future work. 

\section{Activation of causal nonseparability by a side channel}
Since a PM is a special case of an MPM with no partial ordering at all assumed for its nodes, one may think that applying the PM definition to an MPM would be sufficient to establish causal (non)separability. Remarkably, this is not true, as we will now demonstrate by an example. 

Consider an MPM for two parties, $A$ and $B$, where $A$ has two nodes and $B$ a single node. Let the dimensions of all input and output systems of all nodes be equal to 2: $d_{X_i} = 2\quad \forall X\in\mathfrak{N} \,, 0 \leq i < 2n_X-1$. We take the MPM to have the following expression, here written in the Pauli basis: 
\begin{multline}\label{eq:Activable}
    W^{AB}=\\
    \frac{1}{8}\left(\mathds{1}+\frac{1}{\sqrt{2}}\left[\GGB{A_0}{x}\GGB{A_2}{z}\GGB{A_3}{z}\GGB{B_0}{z} + \GGB{A_0}{z}\GGB{A_2}{z}\GGB{B_1}{z}\right]\right)\,.
\end{multline}
To keep the above expression concise, the tensor product between matrices acting on separate systems has been omitted, and unit matrices on subsystems are implied. $\GGB{A_0}{z}\GGB{A_2}{z}\GGB{B_1}{z}$ is then the short form of $\GGB{A_0}{z} \otimes \mathds{1}^{A_1} \otimes \GGB{A_2}{z} \otimes \mathds{1}^{A_3} \otimes \mathds{1}^{B_0} \otimes \GGB{B_1}{z}$ in our notation. 

One can demonstrate that \eqref{eq:Activable} is a valid PM as well as a valid MPM by using Theorem \ref{theo:complement}, see Appendix \ref{app:exemple}. This is actually almost the same matrix as the one given in the \textit{activation} example in Ref. \cite{OG2016}, but the isolated node $\Party{A}{1}{}$ between subsystems $A_0$ and $A_1$ is now getting a non-trivial input instead of output. The matrix \eqref{eq:Activable} actually corresponds to a deterministic quantum comb with the fixed causal order $\Party{A}{2}{}\prec B \prec \Party{A}{1}{}$ as one can verify using Eqs. \eqref{eq:QuComb}. This proves that $W^{AB}$ is a causally separable PM, which admits a physical realisation \cite{Chiribella2009}. However, in this causal realisation node $\Party{A}{2}{}$ is before node $\Party{A}{1}{}$.
\begin{figure}
    \centering
    \includegraphics[height=0.9\linewidth]{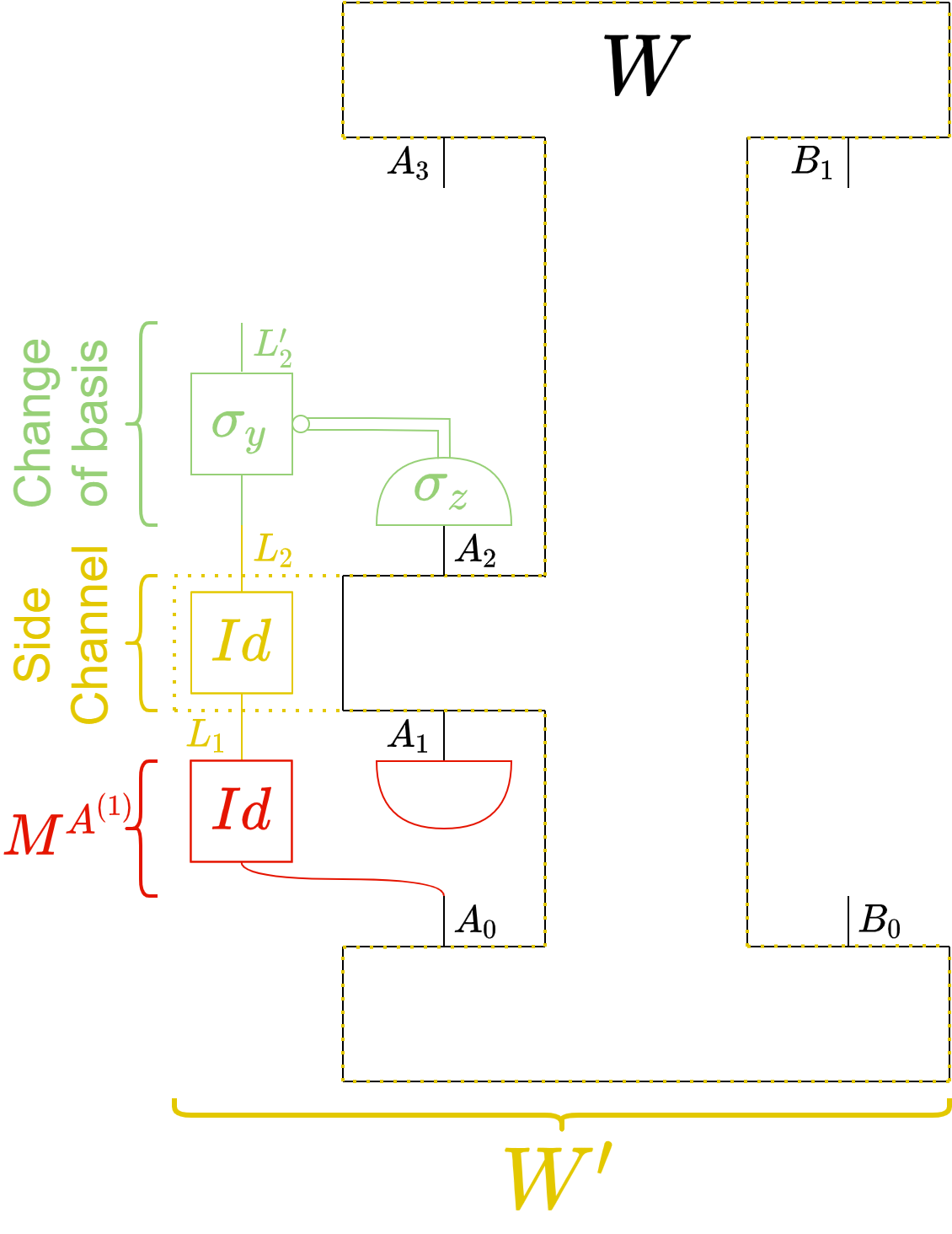}
    \caption{Activation of causal nonseparability by a side channel: process matrix $W$ is causally separable in the PM sense. When extended by a side channel (in yellow), there exists an operation $M^{A^{(1)}}$ (in red) that makes it lose this property between the 2 remaining nodes. By measuring $A_2$ and changing the basis of $L_2$ accordingly (in green), $A$ obtains the OCB PM between the remaining nodes. }
    \label{fig:Activation}
\end{figure}

Yet, if we treat $W$ as an MPM with node $\Party{A}{1}{}$ before node $\Party{A}{2}{}$, Alice can use a side channel to pass on her first input to act on it during her second operation. The extension of $W$ by an identity side channel $Id = \sum_{i,j} \dyad{ii}{jj}^{L_1L_2}$ between the operations of $A$,
\begin{equation}
    W'^{ALB} \equiv W^{AB} \otimes \sum_{i,j} \dyad{ii}{jj}^{L_1L_2}\,,
\end{equation}
is effectively allowed in the MPM formalism; the situation is represented in yellow in Fig. \ref{fig:Activation}. When Alice's first action (in red) is to simply forward her input through the channel, $M^{A^{(1)}}=\sum_{k,l} \dyad{kk}{ll}^{A_0L_1}\otimes \frac{1}{2}\mathds{1}^{A_1}$, the conditional 2-partite (M)PM on the remaining nodes is
\begin{equation}\label{eq:Activated}
    \begin{split}
        &W^{A_2A_3L_2B_0B_1}= \TrX{A_0A_1L_1}{M^{A^{(1)}} \cdot W'^{ALB}} \\
        &=\frac{1}{8}\left(\mathds{1}+\frac{1}{\sqrt{2}}\left[\GGB{A_2}{z}\GGB{A_3}{z}\GGB{L_2}{x}\GGB{B_0}{z} + \GGB{A_2}{z}\GGB{L_2}{z}\GGB{B_1}{z}\right]\right) \,.
    \end{split}
\end{equation}
If $A$ performs a measurement in the $\GGB{}{z}$ basis on $A_2$ and obtains the outcome \enquote{+}, $A$ and $B$ are left with the OCB process matrix $W^{A_3L_2B_0B_1} = \frac{1}{4}\left(\mathds{1}+\frac{1}{\sqrt{2}}\left[\GGB{A_3}{z}\GGB{L_2}{x}\GGB{B_0}{z} + \GGB{L_2}{z}\GGB{B_1}{z}\right]\right)$ on the remaining systems, which can be used to violate a causal inequality as described in Ref. \cite{OCB2012}. If she obtains the outcome \enquote{$-$}, $A$ and $B$ are left with the matrix $W^{A_3L_2B_0B_1} = \frac{1}{4}\left(\mathds{1} - \frac{1}{\sqrt{2}}\left[\GGB{A_3}{z}\GGB{L_2}{x}\GGB{B_0}{z} + \GGB{L_2}{z}\GGB{B_1}{z}\right]\right)$, which differs from the OCB matrix merely by a change of basis. Hence the same noncausal correlations as with the OCB PM can be obtained if $A$'s operations are modified to account for this change of basis. She can achieve this by applying a $\sigma_y$ transformation on $L_2$ controlled by the output of the $\sigma_x$ measurement of $A_2$ (in green). Consequently, the ability for the parties to use side communication with the MPM exhibits non-trivial differences in terms of achievable correlations compared to the PM.
\section{Conclusion}\label{sec:disc}
To summarise, we defined an extension of the PM formalism that allows multiple rounds of information exchange for each party, which we named the \enquote{multi-round process matrix} (MPM). We provided a complete characterisation of the set of valid MPMs in Theorem \ref{theo:complement} using the projective formulation of the validity constraints on deterministic quantum combs. This highlighted a connection between the set of MPMs and deterministic quantum combs: Theorem \ref{theo:PV=UPTildeC} demonstrates that MPMs are affine combination of all the quantum combs compatible with the local ordering of each party. Finally, we motivated a new notion of causal separability specific to MPMs, Definition \ref{def:MPM_CnS}, that takes into account the possibility of using side channels. A non-trivial consequence of this possibility was that the notion of causal nonseparability for PMs and MPMs are not equivalent when applied to the same operator -- we showed an example of an operator that is causally separable when considered as a PM, and causally nonseparable when considered as an MPM.

Several paths can be considered for the continuation of this work. First, the question of finding whether there exists a bound on the dimensions of the side channels needed to certify causal non-separability in the MPM sense has been left open in this work. 

Then, a natural path is to analyse the connection of this work with the axiomatic theory of higher-order quantum computation by Bisio and Perinotti \cite{Perinotti2016,Bisio2018}, which captures the MPM as a specific type of process within the infinite hierarchy of higher-order processes; this will be explored in an upcoming article. Another related work is Jia's correlational approach to quantum theory \cite{jia2020}. Several results derived here were partially found or hinted in these works.

Finding new protocols that provide an advantage over regular quantum communications \cite{Colnaghi2012,Feix2015} constitutes another research direction. We expect that the tool provided by the multi-round process matrix would be useful to this purpose, especially to formulate communication protocols using indefinite causal order as a resource \cite{Taddei_2019,kristjnsson2019resource}. 

Finally, one may expect that the projective methods used here to characterise MPMs could prove useful in the context of various other problems in the field of (M)PMs, as well as in the broader field of higher-order processes, such as for deriving convenient characterisations of more general processes in the hierarchy. Understanding how the notion of causal nonseparability generalises to higher orders is another question of great interest, which could potentially unveil new phenomena specific to these orders.

\begin{acknowledgments}
This work is an adaptation of T. H.'s master's thesis at the {\'E}cole polytechnique de Bruxelles (ULB) that can be found on QuIC website \cite{master}. Illustrations were drawn using \href{https://www.draw.io}{draw.io}. This publication was made possible through the support of the ID\# 61466 grant from the John Templeton Foundation, as part of the “The Quantum Information Structure of Spacetime (QISS)” Project (\href{http://www.qiss.fr/}{qiss.fr}). The opinions expressed in this publication are those of the authors and do not necessarily reflect the views of the John Templeton Foundation. 
This work was supported by the Program of Concerted Research Actions (ARC) of the Universit\'{e} Libre de Bruxelles. O. O. is a Research Associate of the Fonds de la Recherche Scientifique (F.R.S.–FNRS).
\end{acknowledgments}

\bibliographystyle{plainnat}
\bibliography{references}

\appendix

\section{Projective formulation}
\subsection{The projective superoperator \label{app:projo}}
This section aims to introduce the reader to the mapping $ \Dep{X}{} : \LinOp{\Hilb{}} \rightarrow \LinOp{\Hilb{}}$, defined as
\begin{equation}
    \DepPar{X}{\cdot} = \MapX{X}{\cdot} \,, \label{eq:Depolop}
\end{equation}
where $X$ is a tensor factor of $\Hilb{}$, and $d_X$ its dimension. It was introduced in Ref. \cite{Witness} alongside its shorthand prescript notation. To illustrate its properties, we consider an operator $M$ defined on some space $\LinOp{\Hilb{A} \otimes \Hilb{B} \otimes \Hilb{C}}$ of dimension $d^2 = (d_A d_B d_C)^2$. 

We define a multiplication within the prescript that corresponds formally to a composition of several maps as
\begin{equation}
	\MapX{A}{\MapX{B}{M}} \equiv \Dep{AB}{M} \,.
\end{equation}
This \enquote{multiplication} inherits the properties of partial tracing: it possesses an identity element (see below), it is associative as well as commutative as partial tracings on different subsystems are associative and commute, but it does not possess an inverse. Next, note that this mapping is CPTP, as well as idempotent: 
\begin{equation}\label{eq:app:idempot}
	\Dep{AA}{M} = \Dep{A^2}{M} = \Dep{A}{M} \,.
\end{equation}
    
Using the Hilbert-Schmidt product as the inner product on the space of linear operators,
\begin{gather}
	\InProd{\cdot}{\cdot} \,:\, \LinOp{\Hilb{}} \times \LinOp{\Hilb{}} \rightarrow \mathbb{C} \quad , \notag \\
	M,N \mapsto \InProd{M}{N} \equiv \TrX{}{M^\dag \cdot N}\quad,
\end{gather}
one can verify that the mapping is self-adjoint with respect to it. Let $N \in \LinOp{\Hilb{A} \otimes \Hilb{B} \otimes \Hilb{C}}$ be another arbitrary operator defined on the same space as $M$, then
\begin{align}
	\InProd{M}{\Dep{A}{N}} &=\TrX{ABC}{M^\dag \cdot \left( \MapX{A}{N} \right)} \notag \\
	&= \TrX{BC}{ \TrX{A}{M^\dag} \cdot \frac{1}{d_A}\,\TrX{A}{N}} \notag \\
	&= \TrX{ABC}{\left( \MapX{A}{M^\dag} \right) \cdot N}\, ,
\end{align}
where the identity $ \TrX{ABC}{ U^{ABC}\cdot \left(\mathds{1}^A \otimes V^{BC}\right) } = \TrX{BC}{\TrX{A}{U^{ABC}} \cdot V^{BC} }$ have been used to go to the second and third lines.
Thus,
\begin{equation}\label{eq:app:selfadj}
	\InProd{M}{ \Dep{A}{N} } = \InProd{ \Dep{A}{M} }{N}\quad.
\end{equation}
This means that the mapping is idempotent and self-adjoint, making it an orthogonal projector onto a subspace of $\LinOp{\Hilb{A} \otimes \Hilb{B} \otimes \Hilb{C}}$.

We define a linear addition operation within the prescript as 
\begin{equation}
	\Dep{A}{M} + \Dep{B}{M} \equiv \Dep{ A + B}{M} \,,
\end{equation}
and for each element $\DepPar{A}{\cdot}$ we define the inverse element $\DepPar{-A}{\cdot} \equiv - \MapX{A}{\cdot}$, such that $\Dep{A+(-A)}{M} \equiv \Dep{0}{M} = 0$, where $\Dep{0}{}$ is the additive identity element, which corresponds to the zero mapping. We will introduce the minus sign as a shorthand notation $\Dep{A+(-A)}{M} \equiv \Dep{A-A}{M}$ for addition of inverse elements. With this addition and multiplication, the prescripts have the algebraic structure of a ring \cite{herstein1975topics}. One could also promote the structure to an algebra by defining a scalar multiplication the same way the inverse additive element have been defined, \textit{i.e.} $\DepPar{\lambda A}{\cdot} \equiv \lambda\MapX{A}{\cdot}$, for all $\lambda \in \mathbb{C}$. However, note that the idempotency property is not in general conserved under scalar multiplication nor under addition of prescripts. Hence, for the purposes of this article, we restrict the algebra to a ring and we only consider combinations under addition that result in valid projectors, \textit{i.e.} idempotent elements. 

Finally, we define 2 particular elements of the ring. One is the identity mapping $\mathcal{I}$, such that $\mathcal{I}[M] = M$ for all operators, which corresponds to the multiplicative identity element in the ring. In prescript notation it is then denoted as a 1:
\begin{equation}
    \mathcal{I}[M] \equiv \Dep{1}{M}\,.
\end{equation}
The other is the multiplicative absorbing element $\mathcal{D}: \mathcal{D}\Proj{}{} = \Proj{}{}\mathcal{D} = \mathcal{D}$, for every $\Proj{}{}$ in the ring, which is no other than the projector to the subspace spanned by the unit matrix $\frac{1}{d}\mathds{1}$, or equivalently the mapping applied on the full system $\mathcal{D}(M) = \MapX{}{M}$. In our example, $\mathcal{D}$ is written in prescript notation as
\begin{equation}
    \mathcal{D}[M] \equiv \Dep{ABC}{M} \,.
\end{equation}
We will also use superscripts when those projectors are applied on parts of composite subsystems, \textit{e.g.} $\left(\mathcal{I}^A\otimes \mathcal{D}^{BC}\right)[M] \equiv \Dep{1BC}{M} = \Dep{BC}{M}$.
\subsection{Comb conditions in projective formulation \label{app:CombCd}}
In this section, we derive an alternative formulation of the set of necessary and sufficient constraints an operator has to satisfy to be a valid deterministic quantum comb (Eqs. \eqref{eq:QuComb_Chiri}, see \cite[Theorems 3 and 5]{Chiribella2009}). This formulation is in the spirit of what was done in Ref. \cite{Witness} for the PM validity conditions. We will find a recursive way of building the projector onto the subspace of valid quantum combs for an increasing number of teeth. To do so, we use the shorthand notation $M^{(j)} \equiv M^{X_0X_1\cdots X_{2j-1}}$, so that $M^X \equiv M^{(n_X)}$, to refer to a deterministic quantum comb with $j$ teeth. 

The set of conditions to be proven equivalent to \eqref{eq:QuComb_Chiri} are equations \eqref{eq:QuComb} and \eqref{eq:ProjComb}, rewritten here for convenience:%
\begin{subequations}
    \begin{gather}
        M^X\geq0 \,,\label{eq:app:QuComb_positivity}\\
        \mathcal{P}_{n_X}^{X}\left[M^X\right] = M^X \,,\label{eq:app:QuComb_subspace}\\
        \TrX{}{M^X} = \prod_{j=0}^{n_X-1} d_{X_{2j}} \,,\label{eq:app:QuComb_normalisation}
    \end{gather}%
    \label{eq:app:QuComb}%
\end{subequations}%
where $\mathcal{P}_{n_X}^{X}\equiv \mathcal{P}_{n_X}^{\Party{X}{1}{}\prec \ldots \prec \Party{X}{n_X}{}} $ is the projector to the validity subspace of $n_X$-combs, which is given by the recursive relation \eqref{eq:ProjComb}
\begin{equation}\label{eq:PC}
    \begin{split}
        &\mathcal{P}_{0} = \DepPar{1}{\cdot} \,,\\
        &\mathcal{P}_{n_X}^{\Party{X}{1}{}\prec \ldots \prec \Party{X}{n_X-1}{}\prec \Party{X}{n_X}{}} =\\
        &\,\DepPar{\left(1-X_{2n_X-1}\right)}{\cdot} + \DepPar{X_{2n_X-2}X_{2n_X-1}}{\Proj{\Party{X}{1}{}\prec \ldots \prec \Party{X}{n_X-1}{}}{n_X-1}} .
    \end{split}
\end{equation}
Keep in mind that we are using the shorthand notation $\Proj{\Party{X}{1}{}\prec \ldots \prec \Party{X}{n_X-1}{}}{n_X-1} \equiv \Proj{X_0\ldots X_{2n_X-3}}{n_X-1} \otimes \mathcal{I}^{X_{2n_X-2}X_{2n_X-1}}$ to indicate that we are projecting the subsystems $X_0\ldots X_{2n_X-3}$ on a subspace of quantum ($n_X$-1)-combs with causal ordering $\Party{X}{1}{}\prec \ldots \prec \Party{X}{n_X-1}{}$.

Note that condition \eqref{eq:app:QuComb_positivity} is just condition \eqref{eq:QuComb_Chiri_pos}---the condition that a quantum comb is PSD, which reflects the fact that it is the CJ operator of a CP map. Therefore, equivalence between the set of conditions \eqref{eq:QuComb_Chiri} and the set of conditions \eqref{eq:app:QuComb} with \eqref{eq:PC} would follow if we show that \eqref{eq:QuComb_Chiri_subspace} is equivalent to \eqref{eq:app:QuComb_subspace} and \eqref{eq:app:QuComb_normalisation}, with the projector in \eqref{eq:app:QuComb_subspace} given by \eqref{eq:PC}. This is what we show next.

The normalisation condition \eqref{eq:app:QuComb_normalisation} for deterministic combs is obtained directly by taking the total trace of the comb. We can indeed use the identity $\TrX{AB}{M^A \otimes N^B} = \TrX{A}{M^A} \TrX{B}{N^B}$, which allows to nest the set of linear constraints \eqref{eq:QuComb_Chiri_subspace} within one another. As it yields $\TrX{X_{2i-1}X_{2i-2}}{M^{(i)}} = d_{X_{2i-2}} \; M^{(i-1)}$, one can see that the total trace of a valid comb is equal to the product of its input dimensions (reflecting the fact that a deterministic quantum comb correspond to a channel from all inputs to all outputs). Hence, Eq. \eqref{eq:app:QuComb_normalisation} is a necessary condition.


We next show that for a 1-comb, Eq. \eqref{eq:app:QuComb_subspace} with the projector $\Proj{X}{1}$ defined as in \eqref{eq:PC} is necessary. Indeed, the constraint for a 1-comb is equivalent to
\begin{align}
    \TrX{X_1}{M^{(1)}}\,&=\,\mathds{1}^{X_0} & \iff \nonumber\\
    \Dep{X_1}{M^{(1)}} &= \frac{\mathds{1}^{X_0X_1}}{d_{X_0}d_{X_1}} \; d_{X_0}& \overset{\eqref{eq:app:QuComb_normalisation}}{\Longrightarrow} \nonumber\\
    \Dep{X_1}{M^{(1)}} &= \Dep{X_0X_1}{M^{(1)}} &, \label{eq:app:P^1_raw}
\end{align}
where we multiplied by $\frac{\mathds{1}^{X_1}}{d_{X_1}} \otimes$ on both sides between the first and second lines, and we used the definition \eqref{eq:Depolop} on the left-hand side, then we used Eq. \eqref{eq:app:QuComb_normalisation}, \textit{i.e.} $\TrX{}{M^{(1)}}=d_{X_0}$, together with \eqref{eq:Depolop} on the right-hand side to go to from the second line to the last. Hence, \eqref{eq:app:P^1_raw} is necessary.

The relation \eqref{eq:app:P^1_raw} is rephrased as a projector as $M^{(1)} - \Dep{X_1}{M^{(1)}} + \Dep{X_0X_1}{M^{(1)}} = M^{(1)}$ \cite{Dynamics}, or
\begin{equation}\label{eq:app:P^1}
    \Proj{X}{1}\left[M^{(1)}\right] \equiv \Dep{1-X_1+X_0X_1}{M^{(1)}} = M^{(1)} \,,
\end{equation}
where we have defined the 1-comb projector $\Proj{X}{1} \equiv \Proj{\Party{X}{1}{}}{1}$ from the space of operators $\LinOp{\Hilb{X_0}\otimes \left(\Hilb{X_1}\right)^*}$ to the subspace of (nonnormalised) deterministic quantum 1-combs. It is a projector as the relation $\left( 1-X_1+X_0X_1 \right)^2 = 1-X_1+X_0X_1$ follows from the algebraic rules of the ring defined in Sec. \ref{app:projo}. This last equation, \eqref{eq:app:P^1}, is exactly \eqref{eq:app:QuComb_subspace} with \eqref{eq:PC} in the case of 1-combs, proving their necessity in this case. 

To prove that Eqs. \eqref{eq:app:QuComb_subspace} and \eqref{eq:app:QuComb_normalisation} with \eqref{eq:PC} are sufficient to enforce Eqs. \eqref{eq:QuComb_Chiri_subspace} in the 1-comb case, we start from the projective condition,
\begin{multline}
    \Dep{X_1}{M^{(1)}} - \Dep{X_0X_1}{M^{(1)}} = 0  \\
    \iff \\
    \frac{\mathds{1}^{X_1}}{d_{X_1}} \otimes \left(\TrX{X_1}{M^{(1)}}-\TrX{X_1}{\Dep{X_0}{M^{(1)}}}\right) =0 \\
    \iff \\
    \TrX{X_1}{M^{(1)}}-\frac{\mathds{1}^{X_0}}{d_{X_0}}\TrX{X_0X_1}{{M^{(1)}}} =0  \\
    \overset{\eqref{eq:app:QuComb_normalisation}}{\Longrightarrow} \\
    \TrX{X_1}{M^{(1)}} - \mathds{1}^{X_0} = 0 \, ,
\end{multline}

where we successively used the fact that prescripts commute and the distributive property of the tensor product to go to the second line, and \eqref{eq:Depolop} to make the full trace appear in the third line. Injecting \eqref{eq:app:QuComb_normalisation} we reach \eqref{eq:QuComb_Chiri_subspace} in the last line. Therefore,  \eqref{eq:app:QuComb_subspace} and \eqref{eq:app:QuComb_normalisation} with \eqref{eq:PC} imply \eqref{eq:QuComb_Chiri_subspace} for the case of 1-combs. This completes the proof of equivalence in this case. 

Equivalence in the general case as well as a rule for building the $n_X$-comb projector can be proven by induction. Suppose that the reformulation holds up to $n_X-1$ teeth such that, for $2\leq j \leq n_X-1$,
\begin{equation}\label{eq:app:P^nX-1}
    \begin{gathered}
        \TrX{X_{2j-1}}{M^{(j)}} = \mathds{1}^{X_{2j-2}} \otimes M^{(j-1)},\\
        \TrX{X_1}{M^{(1)}} = \mathds{1}^{X_0}\,,\\
        \iff \\
        \Proj{X}{n_X-1}\left[M^{(n_X-1)}\right] = M^{(n_X-1)} \,, \\
        \TrX{}{M^{(n_X-1)}} = \prod_{j=0}^{n_X-2} d_{X_{2j}} \,,
    \end{gathered}
\end{equation}
where $\Proj{X}{n_X-1} \equiv \mathcal{P}_{n_X}^{\Party{X}{1}{}\prec \Party{X}{2}{} \prec \ldots \prec \Party{X}{n_X-1}{}}$ is the $(n_X-1)$-comb projector. To prove that the $n_X$ case of \eqref{eq:app:QuComb} follows from \eqref{eq:QuComb_Chiri}, we start with $\TrX{X_{2n_X-1}}{M^{(n_X)}} =  \mathds{1}^{X_{2n_X-2}} \otimes M^{(n_X-1)}$. We already know that the trace condition will be satisfied, so to find the projective condition we write
\begin{multline}
    \TrX{X_{2n_X-1}}{M^{(n_X)}} =  \mathds{1}^{X_{2n_X-2}} \otimes M^{(n_X-1)} \, ,\\
    \Dep{X_{2n_X-1}}{M^{(n_X)}} =  \frac{\mathds{1}^{X_{2n_X-2}X_{2n_X-1}}}{d_{X_{2n_X-1}}}\otimes M^{(n_X-1)} \, ,\\
    \Dep{X_{2n_X-1}X_{2n_X-2}}{M^{(n_X)}} = \frac{\mathds{1}^{X_{2n_X-2}X_{2n_X-1}}}{d_{X_{2n_X-1}}} \otimes M^{(n_X-1)}  \, ,\\
    \DepPar{X_{2n_X-1}X_{2n_X-2}}{\ProjOn{X}{n_X-1}{M^{(n_X)}}} = \\
     \frac{\mathds{1}^{X_{2n_X-2}X_{2n_X-1}}}{d_{X_{2n_X-1}}} \otimes M^{(n_X-1)}  \, ,
\end{multline}
where we applied successively $\frac{\mathds{1}^{X_{2n_X-1}}}{d_{X_{2n_X-1}}}\otimes $, $\DepPar{X_{2n_X-2}}{\cdot}$, and $\ProjOn{X}{n_X-1}{\cdot}$ on both sides to show that equivalence in the $n_X-1$ case implies the following relation in the $n_X$ case:%
\begin{multline}
    \TrX{X_{2n_X-1}}{M^{(n_X)}} =  \mathds{1}^{X_{2n_X-2}} \otimes M^{(n_X-1)} 	\Longrightarrow \\
    \Dep{X_{2n_X-1}}{M^{(n_X)}} = \DepPar{X_{2n_X-1}X_{2n_X-2}}{\ProjOn{X}{n_X-1}{M^{(n_X)}}}.
\end{multline}
This last relation is actually the recursive definition $\eqref{eq:PC}$ of the projector as it can be rephrased as
\begin{equation}
    \begin{split}
        &M^{(n_X)} - \Dep{X_{2n_X-1}}{M^{(n_X)}} \\
        &+ \DepPar{X_{2n_X-1}X_{2n_X-2}}{\ProjOn{X}{n_X-1}{M^{(n_X)}}} = M^{(n_X)}\,,\\
        &\equiv \ProjOn{\Party{X}{1}{}\prec \ldots \prec \Party{X}{n_X-1}{}\prec \Party{X}{n_X}{}}{n_X}{M^{(n_X)} }= M^{(n_X)} \,.
    \end{split}
\end{equation}
This yields the definition \eqref{eq:PC}.

To prove that this relation indeed defines a projector, notice the pattern in the prescripts of the recursive relation: let $\mathcal{P}_{n_X}^{\Party{X}{1}{}\prec \ldots \prec \Party{X}{n_X-1}{}\prec \Party{X}{n_X}{}} =  1 - B + ABC $, where $A \equiv \Dep{X_{2n_X-2}}{}$ and $B \equiv \Dep{X_{2n_X-1}}{} $ are idempotent elements, while $C \equiv \Proj{\Party{X}{1}{}\prec \ldots \prec \Party{X}{n_X-1}{}}{n_X-1}$. In the case $n_X=1$, $\Dep{C}{} =1$, which is idempotent, in the case $n_X=2$, $\Dep{C}{} = \Proj{\Party{X}{1}{}}{1}$, which has been proven idempotent as well. If we assume $C$ to be idempotent up to $n_X-1$, then for $n_X$:
\begin{equation}
\begin{aligned}
    \left( 1-B+ABC \right)^2 =& \left(1-B\right)^2 + 2 \left(1-B\right)ABC\\
     & \, + \left(ABC \right)^2 \\
    =& 1-B + 0 + ABC \, ,
\end{aligned}
\end{equation}
where we have used the distributive property of multiplication as well as idempotency of all elements to show that $1-B + ABC$ is idempotent as well. This proves that the objects built in equation \eqref{eq:PC} are idempotent for all $n_X$, thus they define projectors.

Hence, by induction we have that conditions \eqref{eq:app:QuComb} with \eqref{eq:PC} are necessary for \eqref{eq:QuComb_Chiri}. Sufficiency is also proven by induction. Suppose \eqref{eq:app:P^nX-1} holds. For the $n_X$ case we have the implication
\begin{multline}
    \ProjOn{\Party{X}{1}{}\prec \ldots \prec \Party{X}{n_X-1}{}\prec \Party{X}{n_X}{}}{n_X}{M^{(n_X)} }= M^{(n_X)}\\
    \iff\\
    \Dep{X_{2n_X-1}}{M^{(n_X)}} - \DepPar{X_{2n_X-1}X_{2n_X-2}}{\ProjOn{X}{n_X-1}{M^{(n_X)}}} =0 \\
    \Longrightarrow \\
    \TrX{X_{2n_X-1}}{M^{(n_X)} - \DepPar{X_{2n_X-2}}{\ProjOn{X}{n_X-1}{M^{(n_X)}}}} =0 ,
\end{multline}
yielding 
\begin{multline} 
\TrX{X_{2n_X-1}}{M^{(n_X)}} =\\
 \frac{\mathds{1}^{X_{2n_X-2}}}{d_{X_{2n_X-2}}} \otimes \TrX{X_{2n_X-1}X_{2n_X-2}}{\ProjOn{X}{n_X-1}{M^{(n_X)}}}\,.
\end{multline}%
Because of equations \eqref{eq:app:P^nX-1}, we must have that%
\begin{equation}
	\TrX{X_{2n_X-1}X_{2n_X-2}}{\ProjOn{X}{n_X-1}{M^{(n_X)}}} \propto M^{(n_X-1)}.
\end{equation}%
Tracing out both sides, we find that the proportionality constant has to be equal to $d_{X_{2n_X-2}}$ in order for condition \eqref{eq:app:QuComb_normalisation} to hold. Thus, $\TrX{X_{2n_X-1}}{M^{(n_X)}} =  \mathds{1}^{X_{2n_X-2}} \otimes M^{(n_X-1)}$ follows from \eqref{eq:app:QuComb}. Adding this condition to the $n_X-1$ other conditions $\TrX{X_{2j-1}}{M^{(j)}} = \mathds{1}^{X_{2j-2}} \otimes M^{(j-1)}$, $2\leq j \leq n_X-1$ and $\TrX{X_1}{M^{(1)}} = \mathds{1}^{X_0}$, that were already assumed in Eqs. \eqref{eq:app:P^nX-1}, gives the $n_X$ conditions \eqref{eq:QuComb_Chiri_subspace} for an deterministic quantum $n_X$-comb. Hence \eqref{eq:app:QuComb} with \eqref{eq:PC} implies \eqref{eq:QuComb_Chiri} by induction, proving sufficiency.

Note that Eq. \eqref{eq:PC} admits an \enquote{unravelled} formulation,
\begin{equation}\label{eq:PCunravelled}
    \Proj{X}{n_X} = \DepPar{\left( 1-X_{2n_X-1} \left( 1-X_{2n_X-2}  \big( \ldots \big( 1-X_{1} \big( 1-X_{0} \big) \big) \ldots \big) \right) \right)}{\cdot} \,,
\end{equation}
shortened in 
\begin{equation}
	\Proj{X}{n_X} = \mathcal{I} - \DepPar{X_{2n_X-1}}{ \mathcal{I}-\DepPar{X_{2n_X-2}}{ \Proj{X}{n_X-1} } } \,.
\end{equation}%
This will be useful in the proof of Theorem \ref{theo:complement} presented in the next section.

\section{Proof of Theorem \ref{theo:complement} \label{app:Complement}} 
As \eqref{eq:MPM_pos} is the same condition as \eqref{eq:pos}, the proof consists in showing that condition \eqref{eq:MPM_norm} is equivalent to conditions \eqref{eq:complProj} and \eqref{eq:norm}. 

The first step of the proof is to notice that since $\frac{\bigotimes_{X \in \mathfrak{N}} \mathds{1}^X}{\prod_{X\in\mathfrak{N}}\prod_{i=0}^{n_X-1} d_{X_{2i+1}}} \equiv \frac{\mathds{1}}{ d_{out}}$ is a valid tensor product of deterministic quantum combs, the normalisation \eqref{eq:MPM_norm} implies that the set $\{W\}$ have fixed trace norm,
\begin{equation}\label{eq:app:W_norm}
    \TrX{}{W\cdot \frac{\mathds{1}}{d_{out}}} = 1 \iff \TrX{}{W} =  d_{out} \,.
\end{equation}
This is condition \eqref{eq:norm}, which is thus necessary. It also implies that $W=\frac{\bigotimes_{X \in \mathfrak{N}} \mathds{1}^X}{\prod_{X\in\mathfrak{N}}\prod_{i=0}^{n_X-1} d_{X_{2i}}} \equiv \frac{\mathds{1}}{d_{in}}$ is an element of the set of valid MPMs since $\TrX{}{\frac{\mathds{1}}{ d_{in}}\cdot \bigotimes_{X\in\mathfrak{N}} M^X} = 1$ for all $M^X$. This actually corresponds to enforcing condition \eqref{eq:QuComb_norm} on each comb in the tensor product.

For proving the necessity of \eqref{eq:complProj} (and subsequently its sufficiency together with \eqref{eq:norm}), we will work in a convenient Hilbert-Schmidt (HS) basis $\{\sigma_i\}$ for an operator space of dimension $d^2$, whose elements satisfy
\begin{equation}\label{eq:HSbasis}
 \begin{gathered}
    \sigma_i^\dag = \sigma_i \, , \forall i \, ,\\
    \sigma_0 = \mathds{1}\, ,\\
    \TrX{}{\sigma_i} = 0 \, , \forall i \neq 0\, ,\\
    \TrX{}{\sigma_i \cdot \sigma_j} = d \, \delta_{i,j}\, ,\, \forall i,j\,,
\end{gathered} 
\end{equation}
where $\delta_{i,j}$ is the Kronecker delta.

An arbitrary operator $O\in \LinOp{\Hilb{X_0} \otimes \Hilb{X_1} \otimes \ldots}$ can always be expressed in a product basis of the kind described above: $O=\sum_{i_{0},i_{1}\ldots} o_{i_{0},i_{1}\ldots} \; \GGB{X_0}{i_{0}}\otimes \GGB{X_1}{i_{1}} \otimes \cdots$, where $\sigma_j^Y$ is the $j$-th element of a basis of subsystem $Y$ that respects \eqref{eq:HSbasis}, while the $o_{i_{0}i_{1}\ldots} \in \mathbb{C} \,, \forall i_0, \forall i_1,\ldots$ are the associated coefficients of the basis expansion, and the sum runs over all indices: $\sum_{i_{0},i_{1}\ldots} \equiv \sum_{i_0=0}^{d_{X_0}^2-1}\sum_{i_1=0}^{d_{X_1}^2-1}\cdots$. Let $\mathbf{i}=(i_{0},i_{1},\ldots)$ be a multi-index, ranging from 0 to $d_X^2-1$, with first element $\mathbf{0}\equiv (0,0,\ldots) $. Each sub-index corresponds to a subsystem  of X, \textit{i.e.} $i_0$ refers to the index of the $X_0$ component, $i_1$ to the one of $X_1$, etc... We use it to shorten the basis expansion: %
\begin{equation}
O=\sum_{i_{0}i_{1}\ldots} o_{i_{0}i_{1}\ldots} \; \GGB{X_0}{i_0}\otimes \GGB{X_1}{i_1} \otimes \ldots \equiv \sum_{\mathbf{i}} o_{\mathbf{i}} \; \GGB{X}{\mathbf{i}}\,.
\end{equation}

The action of the mapping \eqref{eq:Depolop} defined in section \ref{app:projo} on an operator is to remove certain terms in the basis expansion, as it is projecting on a subset of the basis elements of this type of basis. One can see it as if the superoperator was setting its corresponding sub-index in the HS basis to 0: $\Dep{X_0}{O}= \sum_{\mathbf{i}=(i_0=0,i_1,i_2,\ldots)} o_{\mathbf{i}} \; \GGB{X}{\mathbf{i}}$; the notation under the sum is to be understood as \enquote{summing only over all $\mathbf{i}$ for which $i_0=0$}. For example, $\DepPar{X_0}{\sum_{i_0} o_{i_0} \; \GGB{X_0}{i_0}} = o_{0} \; \GGB{X_0}{0}$, $\DepPar{X_0}{\sum_{i_0i_1} o_{i_0 i_1} \; \GGB{X_0}{i_0} \otimes \GGB{X_1}{i_1}} = \sum_{i_1} o_{0i_1} \; \GGB{X_0}{0} \otimes \GGB{X_1}{i_1}$, etc...

Actually, the projector on the comb validity subspace also shares this property, and so does any tensor product of several projectors of this kind (\textit{i.e.} the projector on $\mathrm{Span}\left\{\bigotimes M\right\}$). This is the content of the following lemma.
\begin{lemm} \label{lemm:P=remover}
    The (tensor product of) mapping(s) $\Proj{X}{n_X}$ defined as in Eq. \eqref{eq:PC} is a superoperator projector whose action is to remove certain basis elements in the Hilbert-Schmidt expansion of an arbitrary CJ operator.
\end{lemm}
\begin{proof}
    The proof is based on the following two observations: (a) Let $\Proj{}{}$ be a projector whose action in the HS basis expansion amounts to removing certain types of terms while leaving the others intact. Then, $\mathcal{I} - \Proj{}{}$ is a projector that also has this property. Indeed, it will leave the terms that are removed by $\Proj{}{}$ and will remove those that are left by $\Proj{}{}$.
    (b) Let $\Proj{}{}$ and $\Proj{'}{}$ be two projectors whose action in the basis expansion amounts to removing certain types of terms. Then the product $\Proj{'}{}\Proj{}{}$ is also a projector of this kind. Indeed, $\Proj{}{}$ would first remove the terms it removes, and then $\Proj{'}{}$ would remove those that it removes from what remains after the action of $\Proj{}{}$. Note that this observation implies in particular that if the claimed property holds for $\Proj{X}{}$ and $\Proj{'Y}{}$ that act on separate systems, it also holds for their tensor product $\left(\Proj{X}{}\otimes \mathcal{I}^Y\right)\left(\mathcal{I}^X \otimes \Proj{'Y}{}\right)$, since $\left(\Proj{X}{}\otimes \mathcal{I}^Y\right)$ and $\left(\mathcal{I}^X \otimes \Proj{'Y}{}\right)$ are obviously also projectors with this property.
    
    Consider the \enquote{unravelled} version of Eq. \eqref{eq:PC}, Eq. \eqref{eq:PCunravelled}. We have seen that $\Dep{X_0}{}$ is such a projector, hence $\Dep{1-X_0}{}$ is also a projector of this kind by (a), so is $\Dep{X_1\left(1-X_0\right)}{}$ by (b), and so is $\Dep{1-X_1\left(1-X_0\right)}{}$ by (a) again. This proves the lemma in the case $n_X=1$ (that is, for Eq. \eqref{eq:app:P^1}). The general case follows by induction: if it is true for $\Proj{X}{n_X-1}$, then it is true for $\DepPar{X_{2n_X-2}}{\Proj{X}{n_X-1}}$ by (b), which is true for  $\mathcal{I}-\DepPar{X_{2n_X-2}}{ \Proj{X}{n_X-1} }$ by (a), which is again true for $ \DepPar{X_{2n_X-1}}{ \mathcal{I}-\DepPar{X_{2n_X-2}}{ \Proj{X}{n_X-1} } }$ by (b), and which is true for $\mathcal{I} - \DepPar{X_{2n_X-1}}{ \mathcal{I}-\DepPar{X_{2n_X-2}}{ \Proj{X}{n_X-1} } }$ by (a). This is Eq. \eqref{eq:PCunravelled} for $n_X$, hereby proving the induction.
\end{proof}
For example, in the case of a single party with one operation $O= \sum_{i_0i_1} o_{i_0 i_1} \; \GGB{X_0}{i_0} \otimes \GGB{X_1}{i_1}$, one has $\ProjOn{X}{1}{O} = \Dep{1-X_1+X_0X_1}{O} = o_{00} \; \GGB{X_0}{0}\otimes \GGB{X_1}{0} + \sum_{\mathbf{i}=(i_0,i_1>0)} o_{i_0i_1} \; \GGB{X_0}{i_0}\otimes \GGB{X_1}{i_1}$; the removed terms have indices $\mathbf{i} = (i_0>0,i_1=0)$. 

In the general case for a single party, we will write the action of a projector on the basis expansion as $\ProjOn{X}{n_X}{\sum_{\mathbf{i}} o_{\mathbf{i}}\;\sigma_{\mathbf{i}}^X} = \sum_{\mathbf{i} \in \{\{M^X\}\}} o_{\mathbf{i}}\; \sigma_{\mathbf{i}}^X$, where $\sum_{\mathbf{i} \in \{\{M^X\}\}}$ is to be understood as \enquote{summing only over the basis elements that belong to the subspace in which the set of $n_X$-combs $\{M^X\}$ is defined}. In the 1-node example of above, $\mathbf{i} \in \{\{M^X\}\}$ is thus equivalent to $\mathbf{i} \in \{(0,0)\} \cup \{(i_0,i_1)\}_{i_0=0,i_1>0}^{d^2_{X_0}-1,d^2_{X_1}-1}$. 

We next observe that the subspace spanned by the operators that are a tensor product of valid deterministic combs for the separate parties is the subspace on which the projector $\Proj{}{} \equiv \bigotimes_{X\in\mathfrak{N}} \Proj{X}{n_X}$ projects. Indeed, the action of $\mathcal{P}$ on an operator is to remove all terms in its Hilbert-Schmidt expansion that are not tensor products of terms allowed (left intact) by the local projectors $\Proj{X}{n_X}$. In other words, $\Proj{}{}$ projects on the subspace spanned by the tensor products of all locally allowed basis elements. Since each locally allowed basis element can be expressed as a linear combination of local deterministic combs, every operator in the subspace on which $\Proj{}{}$ projects is of the form $M= \sum_k q_k \left( \bigotimes_X M^X \right)_k$, where $M^X$ is a valid deterministic comb associated with party $X$, and $q_k$ are real coefficients. Conversely, every operator of this form is obviously left invariant by $\mathcal{P}$.
The trace of this last expression gives $\TrX{}{M} = \sum_k q_k \left( \prod_{X\in\mathfrak{N}} d_{X_{in}} \right)_k$, where for each $k$ the terms in parenthesis are exactly the input dimension of the Hilbert space, so they can be factored out of the sum: $\TrX{}{M} = d_{in} \sum_k q_k$. Hence, requiring that the trace is equal to $d_{in} \equiv \prod_{X\in\mathfrak{N}} d_{X_{in}}$ is equivalent to requiring that $\sum_k q_k=1$. Therefore, conditions \eqref{eq:AffineCombs} are equivalent to requiring that $M$ is an affine sum of tensor products of deterministic quantum combs. 

We will now prove the necessity of \eqref{eq:complProj} by using this observation to compute the quasiorthogonal projector explicitly. Let $M$ be an operator on $\LinOp{\Hilb{}}$ that satisfies conditions \eqref{eq:AffineCombs}. It is written as
\begin{equation}
    M= \bigotimes_{X\in \mathfrak{N}} \left( \sum_{\mathbf{i}_X \in \{\{M^X\}\}} m_{\mathbf{i}_X}\; \GGB{X}{\mathbf{i}_X} \right) \equiv \sum_{\mathbf{i}\in \{\{M\}\}} m_{\mathbf{i}}\; \GGB{}{\mathbf{i}} \,,
\end{equation}
where we have introduced a shortened formulation by making the tensor product implicit. Let there be an arbitrary operator $W = \sum_{\mathbf{j}} w_{\mathbf{j}} \sigma_{\mathbf{j}}$, where $\mathbf{j}$ is a multi-index defined analogously to $\mathbf{i}$. Applying the normalisation condition \eqref{eq:MPM_norm} on $W$ is equivalent to normalising it on all affine sums of tensor product of deterministic combs (see main text), hence Eq. \eqref{eq:MPM_norm_affine} gives
\begin{equation}\label{eq:Tr(WM)}
    \TrX{}{W\cdot M} = \sum_{\mathbf{j}} \sum_{\mathbf{i} \in \{\{M\}\}} w_{\mathbf{j}} m_{\mathbf{i}}\; d \;\delta_{{\mathbf{j}},\mathbf{i}} = 1\,.
\end{equation}
Since $\TrX{}{W} = w_{\mathbf{0}}d$ and $\TrX{}{M} = m_{\mathbf{0}}d$, the values of these two coefficients are fixed using Eqs. \eqref{eq:app:QuComb_normalisation} and \eqref{eq:app:W_norm} obtained above. This allows us to write
\begin{equation}
\begin{aligned}
    \sum_{\mathbf{j}} \sum_{\mathbf{i} \in \{\{M\}\}} w_{\mathbf{j}} m_{\mathbf{i}} d \delta_{\mathbf{j},\mathbf{i}} &= \sum_{\mathbf{i} \in \{\{M\}\}} w_{\mathbf{i}} m_{\mathbf{i}} d \\
    &= w_{\mathbf{0}}m_{\mathbf{0}} d +  \sum_{\mathbf{i} \neq \mathbf{0} \in \{\{M\}\}} w_{\mathbf{i}} m_{\mathbf{i}} d \\
    &= 1 + \sum_{\mathbf{i} \neq \mathbf{0} \in \{\{M\}\}} w_{\mathbf{i}} m_{\mathbf{i}} d\,,
\end{aligned}
\end{equation}
turning the normalisation condition into
\begin{equation}\label{eq:app:normWM}
    \sum_{\mathbf{i} \neq \mathbf{0} \in \{\{M\}\}} w_{\mathbf{i}} m_{\mathbf{i}} d = 0\,.
\end{equation} 
As $w_{\mathbf{i}}, m_{\mathbf{i}} \in \mathbb{R} \; \forall \mathbf{i}$ (since $W$ and $M$ are Hermitian), as $d$ is a non-negative integer, and as in general $m_{\mathbf{i}} \neq 0$ for at least one given value of $\mathbf{i} \neq \mathbf{0}$ (otherwise we are back to the case $M= \frac{\mathds{1}}{\prod d_{X_{out}}}$), we are left with two possibilities: either $\sum_{\mathbf{i} \neq \mathbf{0} \in \{\{M\}\}} w_{\mathbf{i}} m_{\mathbf{i}}=0$ for some nonvanishing $w_{\mathbf{i}}$, or $ w_{\mathbf{i}} = 0 \,, \forall \mathbf{i} \neq \mathbf{0} \in \{\{M\}\}$. We will show that only the second possibility is viable.

Following Ref. \cite{OG2016}, one can always construct a valid deterministic quantum comb of the form $M_{\mathbf{a}} = m_{\mathbf{0}} \GGB{}{{\mathbf{0}}} + m_{\mathbf{a}} \GGB{}{\mathbf{a}}$ for $\mathbf{a} \neq {\mathbf{0}} \in \{\{M\}\} $ with a small enough $ m_{\mathbf{a}}$ coefficient so that Eqs. \eqref{eq:QuComb} are satisfied (it is sufficient to set $m_{\mathbf{0}}$ to $\left(\prod d_{X_{out}}\right)^{-1}$ and $m_{\mathbf{a}}^2 \leq m_{\mathbf{0}}^2$). Then, equation \eqref{eq:app:normWM} yields $w_{\mathbf{a}} = 0$. Doing it again with a different valid comb $M_{\mathbf{b}} = m_{\mathbf{0}} \GGB{}{{\mathbf{0}}} + m_{\mathbf{b}} \GGB{}{\mathbf{b}}$ such that $\mathbf{b} \neq \mathbf{a}$ yields a second condition, $w_{\mathbf{b}} = 0$. Repeating this argument for all possible choice of basis element, $\left\{ M_{\mathbf{k}} = m_{\mathbf{0}} \GGB{}{{\mathbf{0}}} + m_{\mathbf{k}} \GGB{}{\mathbf{k}} \right\}_{\mathbf{k}\neq{\mathbf{0}} \in \{\{M\}\}}$, proves the latter possibility: $w_{\mathbf{i}} = 0 \,, \forall \; {\mathbf{i}} \neq {\mathbf{0}} \; \in \{\{M\}\}$. 

Note that the last condition defines a subspace in the space of operators to which $W$ must belong. This subspace is quasiorthogonal to the subspace $\mathrm{Span}\left\{ \bigotimes M\right\}$, meaning that the intersection of the two subspaces is the span of the unit matrix. Noticing that $w_{\mathbf{0}} \sigma_{\mathbf{0}} = \Dep{X}{W}\equiv \mathcal{D}[W]$, one expresses this last statement in a basis-independent formulation:
\begin{multline}\label{eq:W}
   \{W\}: \sum_{\mathbf{i} \in \{\{M\}\}} w_{\mathbf{i}} \GGB{}{\mathbf{i}} = w_{\mathbf{0}} \GGB{}{{\mathbf{0}}} \\
   \iff \left( \Proj{}{} - \mathcal{D}\right)[W]= 0 \,.
\end{multline}
One can check that $\left( \Proj{}{} - \mathcal{D}\right)^2 = \Proj{}{} - \mathcal{D}$, so $\Proj{}{} - \mathcal{D}$ is a projector. Taking the orthogonal complement, one gets
\begin{equation}
    W = \left(\mathcal{I} - \mathcal{P} + \mathcal{D} \right) [W] \,, \label{eq:oui}
\end{equation}
which is the sought formula.

The fact that Eqs. \eqref{eq:complProj} and \eqref{eq:norm} are sufficient for \eqref{eq:MPM_norm_affine} (and therefore \eqref{eq:MPM_norm}) is straightforward to check. Indeed, plugging \eqref{eq:W}, which is equivalent to \eqref{eq:complProj}, together with \eqref{eq:norm} in Eq. \eqref{eq:Tr(WM)}, one verifies that $m_{\mathbf{0}} \frac{d}{d_{in}} = 1$ because Eq. \eqref{eq:AffineCombs_norm} implies that $m_{\mathbf{0}} = 1/d_{out}$. This completes the proof of the theorem.
\section{Proof of Theorem \ref{theo:PV=UPTildeC} \label{app:proof:V=UC}}
The proof will be split among several lemmas. Given a finite set of orthogonal projectors $\Proj{1}{},\Proj{2}{},\ldots$, we call their \textit{intersection} the orthogonal projector $\Proj{}{\Proj{1}{}\cap\Proj{2}{} \cap \ldots}$ that projects on the intersection of the subspaces on which $\Proj{1}{},\Proj{2}{},\ldots$ project. As shown in \cite{Piziak1999}, if the projectors commute, this intersection is given by the \enquote{product} of the projectors (formally, their composition): 
\begin{equation}
    \Proj{}{\Proj{1}{}\cap\Proj{2}{} \cap \ldots} \equiv \Proj{1}{}\cap\Proj{2}{} \cap \ldots = \Proj{1}{}\Proj{2}{}\ldots \quad .
\end{equation}
Let $\Proj{A\prec B}{n_A+n_B}$ be the projector on the space of valid $(n_A+n_B)$-combs formed by composing $n_A$ operations (or \enquote{teeth}) of a party $A$ with the $n_B$ operations of another party $B$ into an overall deterministic $(n_A+n_B)$-comb where the operations of Alice are all before those of Bob (see the second case from the left in Fig. \ref{fig:CombCompo} for a graphical example). Let $\Proj{B \prec A}{n_A+n_B}$ be the same kind of comb projector but where the operations of $B$ are put before those of $A$ (rightmost case in Fig. \ref{fig:CombCompo}). Then the intersection of these two projectors is
\begin{equation}\label{eq:proj_inter}
    \Proj{A\prec B}{n_A+n_B} \cap \Proj{B \prec A}{n_A+n_B} = \Proj{A\prec B}{n_A+n_B}\Proj{B \prec A}{n_A+n_B} \, .
\end{equation}

%
\begin{figure}
    \centering
    \includegraphics[width=.9\linewidth]{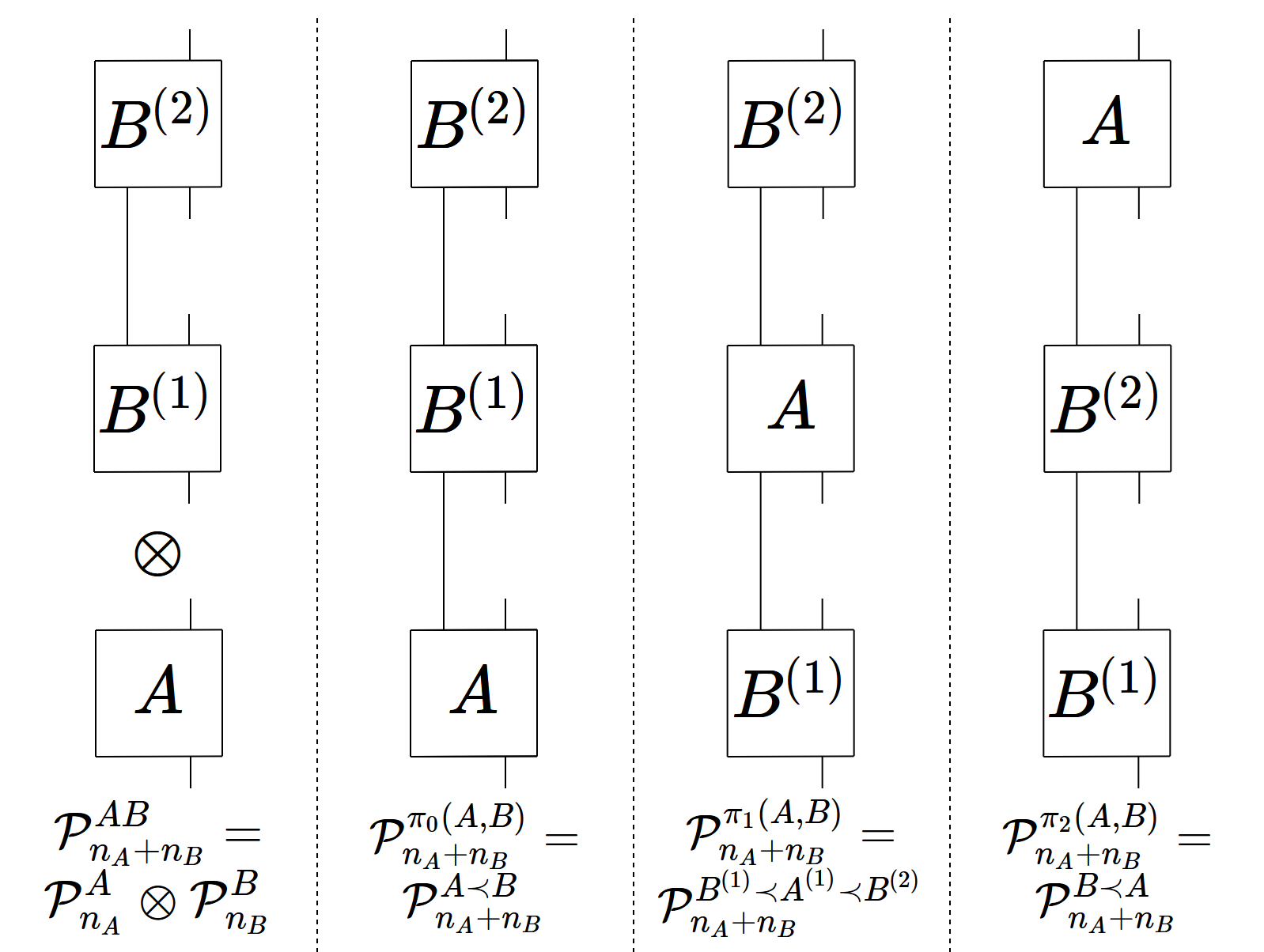}
    \caption{Examples of ways to append $n_A=1$ operations of Alice with $n_B=2$ operations of Bob into valid 3-combs, with the associated projectors on a particular subspace of valid 3-comb written below its graphical representation.}
    \label{fig:CombCompo}
\end{figure}

Any comb in the intersection should therefore be compatible with either all the operations of Alice being first or those of Bob. Intuitively, one can conceive that the only kind of combs valid within this requirement are those where the part of Alice is forbidden to communicate with the part of Bob and vice versa. This is the content of the first lemma.
\begin{lemm} \label{lem:tensorcombs}
    Consider a quantum comb in which some of the teeth are associated with a party $A$ and the others with $B$. Then, the intersection of the subspace of deterministic quantum combs in which the teeth of $A$ are all before those of $B$ (\textit{e.g.} Fig. \ref{fig:CombCompo}, second case from the left) with the one in which the teeth of $B$ are all before those of $A$ (Fig. \ref{fig:CombCompo}, rightmost case) is equivalent to $\mathrm{Span}\left\{ M^A \otimes M^B \right\}$, which is the subspace spanned by the tensor product of a smaller deterministic comb acting on the nodes of $A$ only together with a smaller deterministic comb acting on the nodes of $B$ only (Fig. \ref{fig:CombCompo}, leftmost case). This is because their projectors are equivalent:
    \begin{equation}
        \Proj{A\prec B}{n_A+n_B}\Proj{B \prec A}{n_A+n_B} = \Proj{A}{n_A} \otimes \Proj{B}{n_B} \,.%
        \label{eq:lem:V=UC.1}%
    \end{equation}
\end{lemm}
\begin{proof}
    The proof relies upon the recursive formulation of a quantum comb projector \eqref{eq:PC}. When inspecting formula \eqref{eq:PCunravelled} for $\Proj{A\prec B}{n_A+n_B}$, we see that it is possible to express the projector as: 
    \begin{equation}\label{eq:app:projAB}
        \Proj{A\prec B}{n_A+n_B} = \mathcal{I}^A \otimes \Proj{B}{n_B} - \left(\mathcal{I}^A - \Proj{A}{n_A}\right) \otimes \mathcal{D}^B \,.
    \end{equation}
    Indeed, it holds for the case $n_B=1$:
    \begin{equation}
    \begin{aligned}
        \Proj{A\prec B}{n_A+1} &= \DepPar{1-B_{1}}{\cdot} + \DepPar{B_0B_1}{\Proj{A}{n_A}}\\
        &=\DepPar{1-B_{1}+B_0B_1}{\cdot} - \DepPar{B_0B_1}{1-\Proj{A}{n_A}}\\
        &=\mathcal{I}^A \otimes \Proj{B}{1} - \DepPar{B_0B_1}{\mathcal{I}^A - \Proj{A}{n_A}}\\
        &=\mathcal{I}^A \otimes \Proj{B}{1} - \left(\mathcal{I}^A - \Proj{A}{n_A}\right) \otimes \mathcal{D}^B \,.
    \end{aligned}
    \end{equation}
    Suppose the decomposition \eqref{eq:app:projAB} is true for $n_A+n_B-1$ operations. We define the projector for such a case as $\Proj{A\prec B}{n_A+n_B-1} \equiv \Proj{A\prec \Party{B}{1}{} \prec \ldots \prec \Party{B}{n_B-1}{}}{n_A+n_B-1}$. Then, for $n_A+n_B$:
    \begin{multline}
        \Proj{A\prec B}{n_A+n_B} = \DepPar{1-B_{2n_B-1}}{\cdot} + \DepPar{B_{2n_B-2}B_{2n_B-1}}{\Proj{A\prec B}{n_A+n_B-1}}\\
        = \DepPar{1-B_{2n_B-1}}{\cdot} + \DepPar{B_{2n_B-2}B_{2n_B-1}}{\mathcal{I}^A \otimes \Proj{B}{n_B-1}}\\
        - \DepPar{B_{2n_B-2}B_{2n_B-1}}{\left(\mathcal{I}^A - \Proj{A}{n_A}\right) \otimes \mathcal{D}^B_{n_B-1}}\\
        = \mathcal{I}^A \otimes \Proj{B}{n_B} - \left(\mathcal{I}^A - \Proj{A}{n_A}\right) \otimes \mathcal{D}^B_{n_B}\,,
    \end{multline}
    where we injected $\Proj{A\prec B}{n_A+n_B-1} = \mathcal{I}^A \otimes \Proj{B}{n_B-1} - \left(\mathcal{I}^A - \Proj{A}{n_A}\right) \otimes \mathcal{D}^B_{n_B-1}$ to go from the first equality to the second, with $\mathcal{D}^B_{n_B-1} \equiv \Dep{B_0\ldots B_{2n_B-3}}{}$. Next, we used Eq. \eqref{eq:PC} together with $\DepPar{B_{2n_B-2}B_{2n_B-1}}{\mathcal{D}^B_{n_B-1}} = \mathcal{D}^B$ to go to the last equality. This proves decomposition \eqref{eq:app:projAB} by induction.
    
    Now, applying an analogous decomposition to $\Proj{B\prec A}{n_A+n_B}$, the left-hand side of Eq. \eqref{eq:lem:V=UC.1} becomes
    \begin{multline}
        \Proj{A \prec B}{n_A+n_B}\Proj{B \prec A}{n_A+n_B} =\\
        \left( \mathcal{I}^A \otimes \Proj{B}{n_B} - \left(\mathcal{I}^A - \Proj{A}{n_A}\right) \otimes \mathcal{D}^B \right)\\
        \times \left( \Proj{A}{n_A} \otimes \mathcal{I}^B - \mathcal{D}^A \otimes \left(\mathcal{I}^B - \Proj{B}{n_B}\right) \right)\\
        = \Proj{A}{n_A} \otimes \Proj{B}{n_B}\,. 
    \end{multline}
    To go to the last line, we used the fact that $\mathcal{D}$ is the absorbing (\textit{i.e.} \enquote{zero}) element of the multiplication in the ring (see Sec. \ref{app:projo}), making all the terms where it appears vanish, since $\left(\mathcal{I}^X - \Proj{X}{}\right)\mathcal{D}^X = \left( \mathcal{D}^X - \mathcal{D}^X\right) = 0$. 
\end{proof}
%
The projector of this first lemma is simply the projector on a subspace containing all the $(n_A+n_B)-$combs obtained by taking the tensor product of an $n_A$-comb with an $n_B$-comb (\textit{e.g.} leftmost case in Fig. \ref{fig:CombCompo}). It should also be compatible with permutations of Bob's and Alice's teeth in the full comb that respect the local causal ordering for each party. We will need the $\pi_i(A,B)$ notation introduced in Sec. \ref{sec:MPM=Combs} to refer to the $i$-th valid permutation in the causal ordering of the teeth of parties $A$ and $B$ in order to express this as a corollary.
\begin{coro}\label{coro:V=UC.1}
The subspace made by the intersection of the $(n_A+n_B)$-combs such that all the teeth of $A$ are before those of $B$ with the $(n_A+n_B)$-combs such that all the teeth of $B$ are before those of $A$ is inside the subspace of all $(n_A+n_B)$-combs whose teeth orderings are permutations of the teeth of $A$ and $B$ that respect the local order assumed for $A$ and $B$:   
    \begin{equation}\label{eq:coro:V=UC.1}
        \Proj{A\prec B}{n_A+n_B}\Proj{B \prec A}{n_A+n_B}\Proj{\pi_i(A, B)}{n_A+n_B} = \Proj{A\prec B}{n_A+n_B}\Proj{B \prec A}{n_A+n_B}\, \forall i \,.
    \end{equation}
\end{coro}
\begin{proof}
    The subset of deterministic quantum combs on $\LinOp{\Hilb{A} \otimes \Hilb{B}}$ formed by the tensor product of combs on $A$ and $B$ is automatically a valid comb, no matter how one defines the relative ordering between the teeth of $A$ and $B$ \cite{Chiribella2009}. 
\end{proof}
In other words, this corollary says that when an operator is a tensor product of a deterministic $n_A$-comb with a deterministic $n_B$-comb, then it is a valid deterministic $(n_A+n_B)$-comb for any ordering of the teeth compatible with the partial ordering of the individual combs.

The second lemma needed is just a simplification of the formula for the union of projectors when applied to the case of quasiorthogonal projectors for different teeth orderings. 
\begin{lemm}\label{lemm:UPtilde=1-P+D}
    The union of two arbitrary quasiorthogonal projectors $\CompProj{\pi_i(A,B)}{n_A+n_B}$ and $\CompProj{\pi_j(A,B)}{n_A+n_B}$ is given by%
    \begin{equation}\label{eq:lem:V=UC.2}
        \CompProj{\pi_i(A, B)}{n_A+n_B} \cup \CompProj{\pi_j(A, B)}{n_A+n_B} = \mathcal{I} - \Proj{\pi_i(A, B)}{n_A+n_B} \Proj{\pi_j(A, B)}{n_A+n_B} + \mathcal{D} \,.
    \end{equation}%
\end{lemm}
\begin{proof}
    It is proven by doing the computation explicitly, then simplifying by using again the fact that the projector $\mathcal{D}$ on the span of unity is the zero element. Hence,
	\begin{equation}    
    \begin{aligned}
        \CompProj{\pi_i}{} \CompProj{\pi_j}{} &= \left(\mathcal{I} - \Proj{\pi_i}{} + \mathcal{D}\right) \left(\mathcal{I} - \Proj{\pi_j}{} + \mathcal{D}\right)\\
        &= \mathcal{I} - \Proj{\pi_i}{} - \Proj{\pi_j}{} + \Proj{\pi_i}{}\Proj{\pi_j}{} + \mathcal{D}\,,
    \end{aligned}    
    \end{equation}
    where we have dropped the $n_A+n_B$ subscript and simplified the superscript (\textit{i.e.} $\pi_i(A, B) \equiv \pi_i$ ) in order to shorten the notation. This gives, when plugged into \eqref{eq:UprojCtilde},
    \begin{equation}
    \begin{aligned}
        \CompProj{\pi_i}{} \cup \CompProj{\pi_j}{} &= \mathcal{I} - \Proj{\pi_i}{} + \mathcal{D} + \mathcal{I} - \Proj{\pi_j}{} + \mathcal{D} \\
        & \quad - \left(\mathcal{I} - \Proj{\pi_i}{}  - \Proj{\pi_j}{} + \Proj{\pi_i}{}\Proj{\pi_j}{} + \mathcal{D} \right)\\
        &= \mathcal{I} - \Proj{\pi_i}{} \Proj{\pi_j}{} + \mathcal{D}\,.
    \end{aligned}    
    \end{equation}
\end{proof}

All the elements needed for the main proof are in place. Note that the lemmas have been derived assuming only 2 parties. Actually, these hold analogously for any number of parties, since every formula that was used is associative. In the case where $|\mathfrak{N}|>2$, one just has to group the parties together to get back to the 2-partite scenario.

For the main proof, we start from the right-hand side of equation \eqref{eq:PV=UPTildeC} to which we apply the result of Lemma \ref{lemm:UPtilde=1-P+D}, equation \eqref{eq:lem:V=UC.2}:
\begin{equation}
    \bigcup_{\pi_i} \CompProj{\pi_i(A,B, \ldots)}{n_A+n_B+ \ldots} = \mathcal{I} - \prod_{\pi_i} \Proj{\pi_i(A,B,\ldots)}{n_A+n_B+ \ldots} +\mathcal{D} \, .
\end{equation}
We then apply Corollary \ref{coro:V=UC.1} so the permutations $\pi_i$ of all the nodes are restricted to permutations $\pi_i'$ of whole parties only:
	\begin{equation}
	\begin{aligned}
    \mathcal{I} - \prod_{\pi_i'} \Proj{\pi_i'(A,B, \ldots)}{n_A+n_B+ \ldots} &+\mathcal{D} =\\
    \mathcal{I} - &\Proj{A\prec B\prec \ldots}{n_A+n_B+ \ldots}\Proj{B \prec A\prec \ldots}{n_A+n_B+ \ldots}... +\mathcal{D} \,.
	\end{aligned}    
    \end{equation}
Remember that here a superscript $X$ is referring to $n_X$ nodes like $X \equiv \Party{X}{1}{}\prec\Party{X}{2}{}\prec \ldots \prec \Party{X}{n_X}{}$.
The final step is to apply Lemma \ref{lem:tensorcombs}, equation \eqref{eq:lem:V=UC.1}, yielding
\begin{multline}
    \mathcal{I} - \Proj{A\prec B\prec \ldots}{n_A+n_B+ \ldots}\Proj{B \prec A\prec \ldots}{n_A+n_B+ \ldots}\ldots +\mathcal{D} =\\
    \mathcal{I} - \left(\Proj{A}{n_A}\otimes \Proj{B}{n_B}\otimes \ldots\right) +\mathcal{D} \,.   
\end{multline}
We can now invoke the definition of the quasiorthogonal projector \eqref{eq:complProj} on the right-hand side, and we have proven that 
\begin{equation}
    \bigcup_{\pi_i} \CompProj{\pi_i(A,B, \ldots)}{n_A+n_B+ \ldots} = \prescript{Q}{}{\left(\Proj{A}{n_A}\otimes\Proj{B}{n_B}\otimes\ldots\right)}\quad .
\end{equation}
%
\section{Details for the activation example\label{app:exemple}}
As an illustration of Theorem \ref{theo:complement}, we will now sketch the proof that the operator \eqref{eq:Activable}, rewritten here for convenience
\begin{multline}
    W^{AB}=\\
    \frac{1}{8}\left(\mathds{1}+\frac{1}{\sqrt{2}}\left[\GGB{A_0}{x}\GGB{A_2}{z}\GGB{A_3}{z}\GGB{B_0}{z} + \GGB{A_0}{z}\GGB{A_2}{z}\GGB{B_1}{z}\right]\right)\,,
\end{multline}
is a valid MPM with 3 nodes, where $\Party{A}{1}{}\prec \Party{A}{2}{}$.

This is a PSD matrix with trace $\TrX{}{W^{AB}} = 8 = d_{A_1}d_{A_3}d_{B_1} = d_{out}$, so conditions \eqref{eq:pos} and \eqref{eq:norm} are directly verified. The projective condition \eqref{eq:complProj} is given by 
\begin{align}
    &W^{AB}=\mathcal{I}^{AB} \left[ W^{AB} \right]- \notag\\
    &\DepPar{\left(1-A_3+A_2A_3-A_1A_2A_3+A_0A_1A_2A_3\right)\left(1-B_1+B_0B_1\right)}{W^{AB}} \notag\\ 
    &+ \mathcal{D}^{AB} \left[ W^{AB} \right] \,, \label{eq:WAB}
\end{align}
where \eqref{eq:ProjComb} have been used to find the 2-comb projector on $A$, $\Proj{\Party{A}{1}{}\prec \Party{A}{2}{}}{2}= \Dep{\left(1-A_3+A_2A_3-A_1A_2A_3+A_0A_1A_2A_3\right)}{}$, as well as the 1-comb projector on $B$, $\Proj{B}{1}=\Dep{\left(1-B_1+B_0B_1\right)}{}$.
Notice that $\DepPar{A_3}{W^{AB}} = \frac{1}{8}\left(\mathds{1}+\frac{1}{\sqrt{2}} \GGB{A_0}{z}\GGB{A_2}{z}\GGB{B_1}{z}\right)$, and $\DepPar{B_1}{W^{AB}} = \frac{1}{8}\left(\mathds{1}+\frac{1}{\sqrt{2}}\GGB{A_0}{x}\GGB{A_2}{z}\GGB{A_3}{z}\GGB{B_0}{z} \right)$, while all the other combinations of prescripts are equivalent to the action of $\mathcal{D}^{AB}$ on the matrix, \textit{i.e.} $\DepPar{A_3B_1}{W^{AB}} = \DepPar{B_0B_1}{W^{AB}} = \DepPar{A_2A_3}{W^{AB}}=\ldots = \frac{\mathds{1}}{8} = \mathcal{D}^{AB} \left[ W^{AB} \right]$. This allows us to simplify the projective condition \eqref{eq:WAB} into
\begin{equation}
    W^{AB} = W^{AB} - \DepPar{1-A_3-B_1}{W^{AB}} - \mathcal{D}^{AB} \left[ W^{AB} \right]\,,
\end{equation}
which is effectively verified:
\begin{multline}
    W^{AB} =\DepPar{A_3}{W^{AB}} + \DepPar{B_1}{W^{AB}} - \mathcal{D}^{AB} \left[ W^{AB} \right]\\
    =\frac{\mathds{1}}{4}+\frac{1}{8\sqrt{2}}\left[\GGB{A_0}{x}\GGB{A_2}{z}\GGB{A_3}{z}\GGB{B_0}{z} + \GGB{A_0}{z}\GGB{A_2}{z}\GGB{B_1}{z}\right] - \frac{\mathds{1}}{8}\\
    =\frac{1}{8}\left(\mathds{1}+\frac{1}{\sqrt{2}}\left[\GGB{A_0}{x}\GGB{A_2}{z}\GGB{A_3}{z}\GGB{B_0}{z} + \GGB{A_0}{z}\GGB{A_2}{z}\GGB{B_1}{z}\right]\right).
\end{multline}

The same way, operator \eqref{eq:Activated}, 
\begin{multline}
    W= \\
    \frac{1}{8}\left(\mathds{1}+\frac{1}{\sqrt{2}}\left[\GGB{A_2}{z}\GGB{A_3}{z}\GGB{L_2}{x}\GGB{B_0}{z} + \GGB{A_2}{z}\GGB{L_2}{z}\GGB{B_1}{z}\right]\right) \,,
\end{multline}
written here without superscript, will now be proven to be a valid (M)PM on 2 nodes $A$ and $B$, where $L_2$ is extending the input system of party $A$. It is straightforward to check that it is a PSD matrix with $\TrX{}{W} = 4 = d_{A_3}d_{B_1}$, hence verifying \eqref{eq:pos} and \eqref{eq:norm}. The projective condition \eqref{eq:complProj} is expressed as
\begin{multline}
     W =\\
      \mathcal{I} \left[ W \right]-\DepPar{\left(1-A_3+L_2A_2A_3\right)\left(1-B_1+B_0B_1\right)}{W} +\mathcal{D} \left[ W\right] = \\
     \DepPar{A_3+B_1-A_3B_1+A_3B_0B_1-B_0B_1+L_2A_2A_3B_1-L_2A_2A_3}{W}\,. \label{eq:non}
\end{multline}
Computing each term on the right hand side, we find
\begin{equation}
\begin{aligned}
    &\Dep{A_3}{W} = \frac{1}{8}\left(\mathds{1}+\frac{1}{\sqrt{2}} \GGB{A_2}{z}\GGB{L_2}{z}\GGB{B_1}{z}\right) ,\\
    &\Dep{B_1}{W} = \frac{1}{8}\left(\mathds{1}+\frac{1}{\sqrt{2}}\GGB{A_2}{z}\GGB{A_3}{z}\GGB{L_2}{x}\GGB{B_0}{z} \right) ,\\
    &\Dep{A_3B_1}{W} = \Dep{A_3B_0B_1}{W}  = \Dep{L_2A_2A_3B_1}{W} = \Dep{L_2A_2A_3}{W} =\\
    &\quad\Dep{B_0B_1}{W} = \frac{\mathds{1}}{8}\,.
\end{aligned}
\end{equation}
Therefore, \eqref{eq:non} becomes
\begin{multline}
    W = \Dep{A_3}{W}+\Dep{B_1}{W}-\Dep{A_3B_1}{W}\\
    = \frac{\mathds{1}}{4}+\frac{1}{8\sqrt{2}} \left[\GGB{A_2}{z}\GGB{L_2}{z}\GGB{B_1}{z}+\GGB{A_2}{z}\GGB{A_3}{z}\GGB{L_2}{x}\GGB{B_0}{z}\right]- \frac{\mathds{1}}{8}\\
    = \frac{1}{8}\left(\mathds{1}+\frac{1}{\sqrt{2}}\left[\GGB{A_2}{z}\GGB{A_3}{z}\GGB{L_2}{x}\GGB{B_0}{z} + \GGB{A_2}{z}\GGB{L_2}{z}\GGB{B_1}{z}\right]\right) \,, 
\end{multline}
proving that the projective condition \eqref{eq:complProj} is indeed verified for $W$.

\end{document}